\documentclass[sigconf, nonacm]{acmart}

\usepackage{calrsfs}
\DeclareMathAlphabet{\pazocal}{OMS}{zplm}{m}{n}
\usepackage{epsfig,endnotes,xspace,hyperref,url,diagbox}
\AtBeginDocument{}
\usepackage{amsmath,amsthm}
\usepackage{tikz}
\usepackage{enumerate}
\usepackage{varioref}
\usepackage{graphicx}
\usepackage{epstopdf}
\usepackage{color}
\usepackage{xcolor}
\usepackage{multirow}
\usepackage[font=small, labelformat=simple]{subcaption}

\usepackage{comment}
\usepackage[utf8]{inputenc}
\usepackage{mathtools}
\usepackage{wrapfig}
\usepackage{algorithm}
\usepackage{algpseudocode}

\usepackage[font=small,labelfont=bf]{caption}
\usepackage{rotating} 
\usepackage{lipsum}
\usepackage{slashbox}

\DeclareMathAlphabet{\mathcal}{OMS}{cmsy}{m}{n}

\newtheorem{definition}{Definition}

\newtheorem{theorem}{Theorem}

\newtheorem{lemma}[theorem]{Lemma}

\newcommand{\X}{\mathbf{X}}

\newcommand{\E}{\mathbb{E}}

\DeclareMathOperator*{\argmin}{arg\,min}

\newcommand{\Var}[1]{\ensuremath{\mathrm{Var}\left[\, #1 \,\right]}}
\newcommand{\EV}[1]{\ensuremath{\mathbb{E}\left[\, #1 \,\right]}}

\renewcommand{\Pr}[1]{\ensuremath{\mathsf{Pr}\left[#1\right]}\xspace}

\newcommand{\A}{\mathcal{A}}

\newcommand{\mypara}[1]{\vspace*{0.06in}\noindent\textbf{#1 }}
\newcommand{\mysubpara}[1]{\vspace*{0.06in}\noindent\underline{#1 }}

\newcommand{\aux}{\ensuremath{\mathtt{aux}}\xspace}

\newcommand{\norm}[1]{\ensuremath{ \left\| {#1}\right\|}\xspace}

\newcolumntype{L}{>{\arraybackslash}m{4cm}}

\newcommand{\members}{\ensuremath{\mathcal{M}}}
\newcommand{\centroids}{\ensuremath{\mathbf{C}}}

\newcommand{\cost}[1]{\ensuremath{\mathtt{cost}_{#1}}}
\newcommand{\support}{\ensuremath{\mathtt{Support}}}
\newcommand{\rel}[2]{\ensuremath{_{#1}^{(#2)}}}
\newcommand{\opt}{\ensuremath{\mathtt{OPT}}}
\newcommand{\range}{\ensuremath{\mathbb{R}}}
\newcommand{\polylog}{\ensuremath{\text{ polylog}\xspace}}

\newcommand{\fm}{\ensuremath{\boldsymbol{\alpha}}}
\newcommand{\FM}{\ensuremath{\boldsymbol{A}}}

\newcommand{\cca}{\ensuremath{\mathbf{w'}}}
\newcommand{\grid}{\ensuremath{\mathbf{G}}}
\newcommand{\kmeans}{$k$-means\xspace}
\newcommand{\geometric}{\ensuremath{\mathtt{Geometric}} \xspace}
\newcommand{\uniform}{\ensuremath{\mathtt{Uniform}} \xspace}
\newcommand{\id}{\ensuremath{\mathtt{id}}}
\newcommand{\harmonic}[1]{\ensuremath{\mathtt{HarmonicMean}\left( #1 \right) }\xspace}
\newcommand{\vzeta}{\ensuremath{\boldsymbol{\zeta}}}

\newcommand{\dpfmps}{\ensuremath{\mathsf{DPFMPS}}}
\newcommand{\dpfmpsgen}{\textsf{DPFMPS-Gen}\xspace}
\newcommand{\dpfmpsest}{\textsf{DPFMPS-BasicEst}\xspace}
\newcommand{\dpfmpsimprove}{\textsf{DPFMPS-2PEst}\xspace}

\newcommand{\dpfm}{\ensuremath{\mathtt{DPFM}} \xspace}
\newcommand{\indca}{\ensuremath{\texttt{IND-LAP}}\xspace}
\newcommand{\localclustering}{\textsf{LocCluster}\xspace}
\newcommand{\membershipencoding}{\textsf{MemEnc}\xspace}
\newcommand{\weightestimate}{\textsf{WeightEst}\xspace}
\newcommand{\grr}{\textsf{GRR}\xspace}
\newcommand{\olh}{\textsf{OLH}\xspace}
\newcommand{\ldpca}{\textsf{LDP-AGG}\xspace}
\newcommand{\ldpimprove}{\textsf{LDP-AGG-2PEst}\xspace}

\newcounter{algsubstate}
\renewcommand{\thealgsubstate}{\alph{algsubstate}}
\newenvironment{algsubstates}
  {\setcounter{algsubstate}{0}%
   \renewcommand{\State}{%
     \stepcounter{algsubstate}%
     \Statex {(\thealgsubstate):}\space}}
  {}

\newcommand{\Lapp}[1]{\mathsf{Lap}\left(#1\right)\xspace}

\newtheorem{fact}{Fact}[section]

\newcommand{\changestart}{\begin{color}{blue}}
\newcommand{\changeend}{ ~\!\!\end{color}}

\begin{document}
\title{Differentially Private Vertical Federated Clustering}

\author{Zitao Li}
\affiliation{%
  \institution{Purdue University}
}
\email{li2490@purdue.edu}
\author{Tianhao Wang}
\affiliation{%
  \institution{University of Virginia}
}
\email{tianhao@virginia.edu}
\author{Ninghui Li}
\affiliation{%
  \institution{Purdue University}
}
\email{ninghui@purdue.edu}

\settopmatter{printfolios=true} 

\begin{abstract}
In many applications, multiple parties have private data regarding the same set of users but on disjoint sets of attributes, and a server wants to leverage the data to train a model.  To enable model learning while protecting the privacy of the data subjects, we need vertical federated learning (VFL) techniques, where the data parties share only information for training the model, instead of the private data.  However, it is challenging to ensure that the shared information maintains privacy while learning accurate models. 
To the best of our knowledge, the algorithm proposed in this paper is the first practical solution for differentially private vertical federated \kmeans clustering, where the server can obtain a set of global centers with a provable differential privacy guarantee. 
Our algorithm assumes an untrusted central server that aggregates differentially private local centers and membership encodings from local data parties. 
It builds a weighted grid as the synopsis of the global dataset based on the received information.
Final centers are generated by running any \kmeans algorithm on the weighted grid.
Our approach for grid weight estimation uses a novel, light-weight, and differentially private set intersection cardinality estimation algorithm based on the Flajolet-Martin sketch.   To improve the estimation accuracy in the setting with more than two data parties, we further propose a refined version of the weights estimation algorithm and a parameter tuning strategy to reduce the final \kmeans loss to be close to that in the central private setting.
We provide theoretical utility analysis and experimental evaluation results for the cluster centers computed by our algorithm and show that our approach performs better both theoretically and empirically than the two baselines based on existing techniques.
\end{abstract}

\maketitle
\thispagestyle{empty}

\section{Introduction}

Data privacy laws and regulations such as GDPR~\cite{gdpr} and California Consumer Privacy Act~\cite{caprivact} bring more restrictions and compliance requirements for the data collectors, including the companies and some government agencies.
However, the demand for larger and more comprehensive datasets is increasing as political and business decisions become more and more reliant on different machine learning models.  
In many applications, data about entities are partitioned among multiple data parties and they cannot bring the data together, due to privacy restrictions. 
\emph{Federated learning} (FL) with the \emph{cross-silo} setting~\cite{kairouz2021advances} is a computation concept that can enable these data parties to use their data to train useful models collaboratively without sharing the data. 
But FL by itself cannot provide any provable privacy guarantee in the sense that adversaries can still infer whether one user's data is in the training set (i.e., membership attack~\cite{dwork2017exposed, shokri2017membership, li2021membership}) or even recover the training data (i.e., reconstruction attack~\cite{dwork2017exposed, carlini2019secret, zhang2020secret}) by examining the shared information from local data parties.
As a result, FL needs to be deployed with other privacy techniques,  such as those for satisfying \emph{differential privacy} (DP)~\cite{DMNS06}, to provide provable privacy guarantees.

This paper focuses on an important federated learning setting, \emph{vertical federated learning} (VFL).
Its difference from horizontal federated learning (HFL) is that all parties have data from the {\it same set of users, but their data attributes are different from each other}, while HFL assumes that all the data parties have data from different sets of users but all local datasets have the same attributes~\cite{mcmahan2018dp-rnn, wu2020value, mcmahan2017dpfedavg, wei2020federated}.
VFL has been an interesting topic in the research area since the early 2000s \cite{vaidya2003privacy, vaidya2005privacy, wu13pivot, yunhong2009privacy, gupta2018distributed}.
The papers are usually motivated by medical or financial use cases, where the users' private data are not allowed to be shared between data parties.
More recently, VFL has been adapted by some fintech companies for more real-world services.
For example, WeBank demonstrates how they do risk-control for car insurance cooperating with car rental companies with VFL techniques~\cite{webank}.
Compared with HFL, VFL tasks usually consider fewer data parties. 

How to perform VFL while not leaking private information has been an interesting topic in the security and privacy community~\cite{DN04,vaidya2003privacy, vaidya2005privacy}.
Many existing VFL approaches are based on secure multiparty computation (SMC), including learning classification tree models~\cite{wu13pivot, liu2020federated-forest, vaidya2005privacy}, regression models ~\cite{gu2020federated} and clustering models~\cite{vaidya2003privacy}.
However, the SMC-based methods' final results cannot provide provable resistance to membership or reconstruction attacks, and they usually have high computation and communication overheads.
Other literature employs DP as the security notion to provide resistance to those attacks.
More recently, researchers have developed VFL algorithms with DP guarantee for matrix factorization~\cite{li2022vldb}, regression~\cite{wang2020hybrid}, and boosting model~\cite{chen2020vafl}.
We employ DP as the privacy notion for VFL clustering problem in this paper.

Many problems are more challenging in the VFL setting than in the central setting and the HFL setting.
One example is the \kmeans clustering problem, in which desired solutions minimize the distances between user data points and their closest cluster centers.  
The \kmeans algorithms developed for the central DP setting~\cite{su2016kmeans, ghazi2020dpkmean} require access to all dimensions of data points to compute distances for updating cluster centers. 
In the HFL setting, each data party also has all dimensions of some data points, and can thus compute these distances.  
In the VFL setting, however, each data party has access to only a subset of features.  As we assume there is only an untrusted server for information aggregation, each data party wants to protect its private user data, and aligning each user's record across different parties under the DP privacy constraint is hard.
Thus, the challenge is to \emph{design a differentially private algorithm in which data parties share messages to convey the necessary information for deriving the final global centers.}
There are two expectations regarding the shared messages.
1) The messages satisfy DP and convey local information precisely, even with a small privacy budget.
Note that a user's information is spread among different parties in VFL and the privacy budget needs to be split among all data parties.
2) The messages not only contain synopses of local information, but also can be ``composed'' by the central server to reconstruct correlations of inter-party features.
The synopses of central DP \kmeans algorithms~\cite{su2016kmeans, ghazi2020dpkmean} do not have such a ``composability'' property and the correlations between the inter-party attributes are lost from the central server's view.
We show that correlation retaining is essential for the VFL \kmeans problem, and the accuracy of the estimated correlations largely affects the final cost of \kmeans problem in our experiments.

\mypara{Our contributions.}
This paper proposes a solution for the differentially private vertical federated \kmeans with multiple data parties and an untrusted central server.
All the information shared by data parties in the process, as well as the final result, satisfy DP.
The key idea is to have each data party generate a differentially private ``data synopsis'', including the {\it partial centers} and encoded {\it membership information} that describes local \kmeans results based on its partial view.
The server then runs a central \kmeans on the Cartesian product of all partial centers considering their weights, where the weight for each joint center is an estimate of the cardinality of intersection among users belonging to each partial center.
Our main contributions are summarized as follows:

\mypara{$\bullet$}
We propose the first (according to our knowledge) differentially private VFL \kmeans algorithm with an untrusted central server. 
We let each data party encode its memberships of the local clusters into Flajolet-Martin (FM) sketches and take advantage of the parallel composition property of DP to reduce the amount of noise (Algorithm~\ref{algo:dpfmps}). 
Because FM sketches support only union operations while we need intersection operations, we design an algorithm (Algorithm~\ref{algo:dpfms-est}) with inclusion-exclusion rules for the server to estimate the intersection cardinalities of memberships.
We also prove a theoretical utility guarantee for the final global $k$ centers derived by the server with limited computation and communication overhead.

\mypara{$\bullet$}
The cardinality estimation errors can grow very fast when the number of data parties in VFL increases.
To improve the estimation accuracy when more than two data parties are involved, we propose a heuristic estimation algorithm (Algorithm~\ref{algo:multiparty}).
It estimates the intersection cardinalities of memberships from all parties based on pair-wise intersection cardinalities and reduces estimation errors significantly.
In addition, we propose a heuristic method to choose the local clustering parameter for smaller final losses.

\mypara{$\bullet$}
Our experiments show that our proposed methods can outperform the other baseline methods and even approach the non-private VFL \kmeans algorithm when sufficient users are in the dataset.
We also conduct ablation studies to empirically demonstrate the impact and effectiveness of each component of our algorithm. 

\mypara{Roadmap.}
We revisit the necessary background information in Section~\ref{sec:background}; we give an overview of the VFL clustering problem and our solution in Section~\ref{sec:algo}, and provide more details of the key components in Section~\ref{sec:membership} and \ref{sec:improve}; experimental results are shown in Section~\ref{sec:experiments}; Section~\ref{sec:related} discusses the related work from different perspectives, followed by a conclusion in Section~\ref{sec:conclustion}.
Because of the space limitation, the proofs and additional experimental results are provided in the appendix of the full version~\cite{manuscript}.
\vspace{-0.2cm}
\section{Background}
\label{sec:background}

\subsection{Differential Privacy}

\begin{definition}[Differential privacy \cite{DMNS06}]
A randomized algorithm $\A$ is $(\epsilon, \delta)$-differentially private if for any pair of datasets $\X$, $\X'$ that differ in one record and for all possible subset $O$ of possible outputs of algorithm $\A$,
$ \Pr{\A(\X) \in O} \leq e^\epsilon \Pr{\A(\X') \in O} + \delta .$
\end{definition}

Three properties of DP are frequently used to build complicated algorithms.  
Assume that there are two subroutines $\A_1(\cdot)$ and $\A_2(\cdot)$ that can provides $(\epsilon_1, \delta_1), (\epsilon_2, \delta_2)$-DP protection.
\emph{Sequential composition} states that $\A_1(\X, \A_2(\X))$ satisfies $(\epsilon_1 + \epsilon_2, \delta_1+\delta_2)$-DP.
On the other hand, \emph{parallel composition} states that combining two subroutines each only accessing a non-overlapping sub-dataset $\X_1$ or $\X_2$ satisfies $(\max\{\epsilon_1, \epsilon_2\}, \max\{\delta_1, \delta_2\})$-DP.
A third property, \emph{post-processing} property states that, any data independent operation on an $(\epsilon, \delta)$-DP algorithm's result still satisfies $(\epsilon, \delta)$-DP. 

\mypara{Laplace mechanism.}
One of the most classic DP mechanisms, Laplace mechanism, adds Laplace noise to the return of a function $f$ to ensure the result is differentially private.
The variance of the noise depends on $\mathsf{GS}_f$, the \emph{global sensitivity} or the $L_1$ sensitivity of $f$, defined with a pair of neighboring datasets as,
$
\mathsf{GS}_f = \max\limits_{\X \simeq \X'} ||f(\X) - f(\X')||_1.
$
The Laplace mechanism mechanism $\A$ is formalized as $ \A_f(\X) =f(\X) + \Lapp{\frac{\mathsf{GS}_f}{\epsilon}}$, 
where $\Lapp{b}$ denotes a random variable sampled from the zero-mean Laplace distribution with scale $b$.
When $f$ outputs a vector, $\A$ adds independent samples of $\Lapp{\frac{\mathsf{GS}_f}{\epsilon}}$ to each element of the vector.

\mypara{Tighter DP sequential composition.}
The notion of R\'enyi Differential Privacy (RDP)~\cite{mironov2017renyi} provides a succinct way to track the privacy loss from a composition of multiple mechanisms by representing privacy guarantees through moments of privacy loss.

\begin{definition}[R\'enyi Differential Privacy~\cite{mironov2017renyi}]
A mechanism $\mathcal{M}: \mathcal{X} \rightarrow \mathcal{Y}$ is said to satisfy $(\nu, \tau)$-RDP if the following holds for any two neighboring datasets $\X, \X^\prime$ 
\begin{align*}
    \frac{1}{\nu-1}\log \E_{o\sim \A(\X)} \left[ \left( \frac{\Pr{\A(\X) = o}}{\Pr{\A(\X') = o}}\right)^{\nu} \right] \leq \tau.
\end{align*}
\end{definition}
\begin{fact}[RDP Sequential Composition~\cite{mironov2017renyi}]
    If $\A_1$ and $\A_2$ are $(\nu, \tau_1)$-RDP and $(\nu, \tau_2)$-RDP respectively then the mechanism combining the two $g(\A_1(\X), \A_2(\X))$ is $(\nu, \tau_1 + \tau_2)$-RDP.
\end{fact}

\begin{fact}[RDP to $(\epsilon, \delta)$-DP~\cite{mironov2017renyi}]
    If a mechanism is $(\nu, \tau)$-RDP, then it also satisfies $(\tau + \frac{\log 1/\delta}{\nu - 1}, \delta)$-DP.
\end{fact}

With the sequential composition of RDP and the conversion to $(\epsilon, \delta)$-DP, the privacy loss of $M$ sequential mechanism can be improved from $O(M\epsilon)$ to the order of $O(\sqrt{M} \epsilon)$.

\subsection{\kmeans Clustering}

The \kmeans problem~\cite{macqueen1967kmeans} is one of the most well-known clustering problems. 
With a parameter $k$ and a dataset $\X \in \mathbb{R}^{n\times m}$, the goal of the problem is to output a set of $k$ centers $\centroids$ that can minimize the distance of data points to the nearest centers.
The cost (or loss function) is formalized as
$\cost{\X}(\centroids) \coloneqq \sum_{x \in \X} (\min_{c\in \centroids}\norm{x - c}_2^2)$.

The cost function can be extended to weighted data sets, where each data point $x$ has a weight $w(x)$ associated with it.
It is equivalent to the scenario having $w(x)$ copies of the same data point $x$ in $\X$.
The cost becomes $\cost{\X}(\centroids) \coloneqq \sum_{x \in \X} w(x) \cdot (\min_{c\in \centroids}\norm{x - c}^2)$.

Theoretically, there is always a set of optimal $k$ centers and the cost is denoted as $\opt_{\X}^k = \min_{|\centroids| = k} \cost{\X}(\centroids)$.
However, finding the optimal set of centers is NP-hard~\cite{aloise2009nphard}.  
Research interests usually fall on approximate algorithms with polynomial running time. 
For example, the most well-known algorithm, Lloyd’s algorithm~\cite{hartigan1979lloyd} has time complexity $O(nmk)$.
A notation, $(\beta, \lambda)$-approximate, is used to describe the utility guarantee of an approximate algorithm, such that $\cost{\X}(\centroids) \leq \beta \cdot \opt^{k}_{\X} + \lambda$ with any $\X$ and $k$, where $\beta$ is called \emph{approximate ratio}.
The best known non-private algorithm has $\lambda = 0$ and $\beta=1+\eta$ for any fixed $\eta > 0$ when $k$ is a constant \cite{matouvsek2000approximate}; 
but it is unavoidable for DP \kmeans to have $\lambda > 0$~\cite{ghazi2020dpkmean}.

\subsection{Cardinality Estimation Sketches}
Sketches usually refer to a family of succinct data structures that can store some basic information about a large amount of data with very low space and time complexity. 
One of the most well-known sketches is the Flajolet-Martin (FM) sketch~\cite{flajolet1985}, which is designed to estimate the cardinality (i.e., the number of distinct elements) of a (multi)set $\members$.
In FM sketch, all the elements in $\members$ are hashed with $H_\zeta(\cdot)$, an ideal geometric-value hash function.
The estimate of the cardinality is $(1+\gamma)^{\alpha}$, where $\alpha = \max\{H_\zeta(x) | x \in \members \}$ and $\gamma$ is the parameter of hash function.
Typically, multiple (e.g., $1000$) hash functions ($H$ with different hash keys $\zeta$) are used, and we take the harmonic/geometric average of those maximums as the final $\alpha$.
One appealing advantage of FM sketch is that it is mergeable. 
With the same hash key, sketches from different (multi)sets can be merged by taking the maximum, and we can derive the estimate of the cardinality of the union of those (multi)sets.
With this property, we can estimate the cardinalities of the union/intersection of the set in the federated setting without leaking private information.

A recent series of research results show that if the cardinality is large enough, a family of hash-based, order-invariant sketches, including FM sketch, can satisfy DP without adding any additional noise~\cite{smith2020fmsketch, hu2021ca, dickens2022allsketch}.  We will introduce more details in Section~\ref{subsec:psi}.

\section{Overview of Problem and Approach} \label{sec:algo}
In this section, we define the problem of differentially private $k$-means under vertical federated leaning (VFL), provide an overview of our four-phase approach, and discuss the first phase solution.  
The problem of VFL \kmeans (without DP) has been studied before by Ding et al.~\cite{ding2016k}.  
We thus describe the approach in~\cite{ding2016k}, the new challenges when we need to satisfy DP, and our framework.

\subsection{Problem Formulation}

We formalize the VFL $k$-means clustering as the following. 

\mypara{Vertical federated learning (VFL).}
Federated learning~\cite{kairouz2021advances} focuses on learning tasks among multiple data parties without directly sharing their local data. 
VFL assumes that each data party's data are with different features of the same set of users. 
Consider a global view of dataset $\X$, where each row corresponds to a user, and each column corresponds to a feature.  
The setting of VFL is that $\X$ is vertically split into $\X=[\X\rel{}{1} | \ldots, | \X\rel{}{S}]$, so that each data party $\ell\in [S]$ has a local dataset $\X\rel{}{\ell}$ with $m\rel{}{\ell}$ features. 
We assume each user is labeled with a unique $\id$ (e.g., MAC address) and is consistent across all the data parties.

\mypara{Security model.}
We assume that an untrusted central server orchestrates the process and derives the results.  
Different from the local setting of DP, we assume that the data parties have common interests in protecting their data privacy, and none of them collude with the server.
We want to ensure that all the information shared by data parties is differentially private to the central server.
Thus, the server cannot learn any private information, even if it is malicious.
However, to ensure the usefulness of the final output, we need to assume that the server does not deviate from the algorithm.

\mypara{Goal.}
All the $S$ parties want to cooperatively generate differentially private $k$ centers $\centroids$ in the full domain (with all attributes) that can approximately minimize the \kmeans loss $\cost{\X}(\centroids)$.
The challenge is that each party only has its local view (a few attributes), where the final centers are computed with a full view (all attributes).

\subsection{A Non-Private Baseline} \label{subsec:nonprivate}
Ding et al.~\cite{ding2016k} consider VFL \kmeans, and aim to avoid the data communication cost of sending all data to a server.  
A natural approach is thus first to construct a global approximation of the data points and then perform \kmeans clustering on the approximation.  In the approach taken in~\cite{ding2016k}, each data party first finds local cluster centers, then reports to the server these local cluster centers together with which local cluster each data point belongs to.  
The central server can assemble the local centers and local clustering memberships to create a set of weighted pseudo data points of the full dataset. 
We provide more details below.

Each party $\ell$ performs clustering to find $k'$ local cluster centers $\centroids\rel{}{\ell} = \left\{c\rel{1}{\ell}, \ldots, c\rel{k'}{\ell} \right\}$; and then sends the $k'$ centers together with membership information $\mathbf{I}\rel{}{\ell} = \left\{\members\rel{1}{\ell}, \ldots, \members\rel{k'}{\ell} \right\}$ to the central server, where $\members\rel{a}{\ell} = \left\{ \id \mid a = \argmin \norm{x\rel{\id}{\ell} - c\rel{a}{\ell}}^2_2 \right\}$.  
The server constructs $(k')^S$ pseudo data points as a grid from the Cartesian product of the local centers received from $S$ parties, i.e., $\grid =\{(c\rel{a_1}{1}, \ldots, c\rel{a_S}{S})\mid \forall( a_1, \ldots, a_S) \in [k']^S \}$;
and assigns the cardinality of the intersection of the corresponding clusters as  weights to them such that  $w(\grid_{(a_1, \ldots, a_S)}) = \left|\members\rel{a_1}{1}\cap \ldots \cap \members\rel{a_S}{S} \right|$.

This algorithm with only one round of communication can perform well because the grid built by the central server actually maintains most of the necessary information about the local datasets: each local center $c\rel{a_\ell}{\ell}$ is the exact average of the user data in $\members\rel{a_\ell}{\ell}$; moreover, data points in the intersection $\members\rel{a_1}{1}\cap \ldots \cap \members\rel{a_S}{S}$ are expected to distribute around the pseudo point $(c\rel{a_1}{1}, \ldots, c\rel{a_S}{S})$.  If the intersection has a small cardinality or even is an empty set, we can know that the  pseudo point can be ignored.
The weighted grid is similar to a useful data synopsis in the \kmeans cost analysis, called \emph{coreset} \cite{har2004coresets}, which approximates the original dataset information.
As long as the weighted grid nodes are representative enough for a subset of points, the central server can find final centers without accessing the distributed datasets.

\subsection{Challenge in the Privacy-preserving Setting}
\label{subsec:challenge}

The approach described in Section~\ref{subsec:nonprivate} does not consider the privacy leakage problem, as the local cluster centers and membership information sent to the server contain sensitive information.  
To protect users' private information, all the information sent to the central server, including (1) local cluster centers $\centroids\rel{}{\ell}$ and (2) local membership information, should be differentially private.  

For the clustering centers, there already exists comprehensive research of the \kmeans algorithm in the central DP setting~\cite{stemmer2018differentially, nissim2018clustering, huang2018optimal, blum2005practical, nissim2007smooth, feldman2009private, wang2015differentially, nissim2016locating}. 
Thus we can choose a method that works well.

Sending local membership information while satisfying DP is, however, very challenging.
To the best of our knowledge, there is no effective DP algorithm for sharing membership information, especially in the scenario with more than two parties.  Fortunately, the reason that we need to share the membership information is to estimate the weights of each pseudo data point.  
Thus, we do not need to share precise membership information, and just need a private way to {\it estimate the cardinality of the intersection among multiple parties}.  The main technical contribution of this paper is a solution to this problem, which we will described in Section~\ref{sec:membership}.  Our proposed approach leverages DP FM sketch and is extended to support the intersection operation among parties.

\subsection{The Overall Framework}
Figure~\ref{fig:vfl-kmeans} is a visualized workflow of Algorithm~\ref{algo:vfc} with two data parties (note that our algorithm/analysis work with the general case of multiple parties). 
Algorithm~\ref{algo:vfc} consists of four phases: 

\mypara{Phase 1:} Each party clusters local data and generates differentially private local centers (sub-procedure $\localclustering$).

\mypara{Phase 2: } Each party encodes the differentially private ``membership information'' of each local cluster with the private centers and user data points (sub-procedure \membershipencoding).

\mypara{Phase 3:}
    The central server first randomly queries a party for an estimate of the total number of users with the Laplace mechanism and a small privacy budget $\epsilon_0$\footnote{We set $\epsilon_0=0.02\epsilon$ unless we specify in the following text.}. 
    Then the central server receives the private local clustering centers and local membership information of the local clusters.
    It builds a weighted grid, where grid nodes are the Cartesian product of different parties' local centers, and have the estimate of intersection cardinality of the corresponding clusters as their weights (Line 3(c) and sub-procedure \textsf{WeightEstimate}).
    
\mypara{Phase 4:} The central server runs a known central \kmeans algorithm on the weighted grid to generate the final $k$ centers. 

\begin{algorithm}
\caption{Private Vertical Federated Clustering}\label{algo:vfc}
\begin{algorithmic}[1]
\Require Local datasets $\{\X\rel{}{\ell} \in \range^{n \times m^{(\ell)}} \mid \ell \in [S] \}$, total privacy budget $(\epsilon, \delta)$ is divided as $\epsilon_0 =(1-b)\epsilon$, $\epsilon_1 = \frac{b\epsilon}{2S}$,  $\epsilon_2 = \frac{b\epsilon}{2S}, \delta_2 = \frac{\delta}{S}$ for each data party, $k$ and $k'$ for clustering, and auxiliary membership encoding parameters $\aux$.

\Ensure A set of $k$ centers $\{c_1, \ldots, c_k\} \in \range^{k \times m}$

\State Each data party $\ell \in [S]$: 
    \begin{algsubstates}
    \State $\{c\rel{1}{\ell}, \ldots, c\rel{k'}{\ell}\} \leftarrow \localclustering(\X\rel{}{\ell}, \epsilon_1, k')$
    \end{algsubstates}

\State Each data party $\ell \in [S]$: 
    \begin{algsubstates}
    \State $\mathbf{I}\rel{}{\ell}\leftarrow \membershipencoding(\X\rel{}{\ell}, \{c\rel{1}{\ell}, \ldots, c\rel{k'}{\ell}\}, \epsilon_2, \delta_2, \aux)$
    \State sends $\centroids\rel{}{\ell} = \{c\rel{1}{\ell}, \ldots, c\rel{k'}{\ell}\}$ and $\mathbf{I}\rel{}{\ell}$ to server
    \end{algsubstates}

\State Central server:
    \begin{algsubstates}
    \State uses $\epsilon_0$ to estimate the total number of user $\hat{n}$
    \State receives $\{ (\centroids\rel{}{\ell}, \mathbf{I}\rel{}{\ell} ) | \ell \in [S] \}$ from all parties
    \State computes grid by Cartesian product $\grid \leftarrow \centroids\rel{}{1} \times \ldots \times \centroids\rel{}{S}$
    \State computes $w(\grid) \leftarrow \weightestimate(\hat{n}, \epsilon_2, \delta_2, \{\mathbf{I}\rel{}{\ell}| \ell \in [S] \})$
    \end{algsubstates}
\State Central server: 
    \begin{algsubstates}
    \State computes and outputs $\{c_1, \ldots, c_k\} \leftarrow \text{\kmeans}(G, w(G), k)$
    \end{algsubstates}
\end{algorithmic}
\end{algorithm}


In what follows, Section~\ref{sec:local-clustering} describes our approach for Phase 1, and Section~\ref{sec:membership} describes our approach for Phase 2 and 3.


\subsection{Private Local Clustering}
 \label{sec:local-clustering}
We review some existing solutions to generate private centers in the central setting and then explain our adaptation to our VFL setting.

\mypara{DPLloyd.} 
A straight-forward differentially private central \kmeans is the \emph{DPLloyd}~\cite{blum2005practical, su2016kmeans}.
In each iteration, the assignment step is the same as the non-private Lloyd algorithm, where each data point is assigned to the closest center produced from the previous iteration.
The updating step ensures DP by 1) using the Laplace mechanism with sensitivity 1 to get the noisy count of data points assigned to the center, 2) using the Laplace mechanism with sensitivity $r$ (it requires that all attributes are bounded in $[-r, r]$) and a split privacy budget for each dimension to get the noisy sums of the data points assigned to the same center.
The centers are updated as the averages of all data points in the same cluster with the noisy count and noisy sum.
Every iteration consumes privacy budget for computing noisy sums and noisy counts.

\begin{figure*}
    \centering
    \hfill
    \begin{minipage}{0.58\textwidth}
        \centering
        \includegraphics[width=0.99\textwidth]{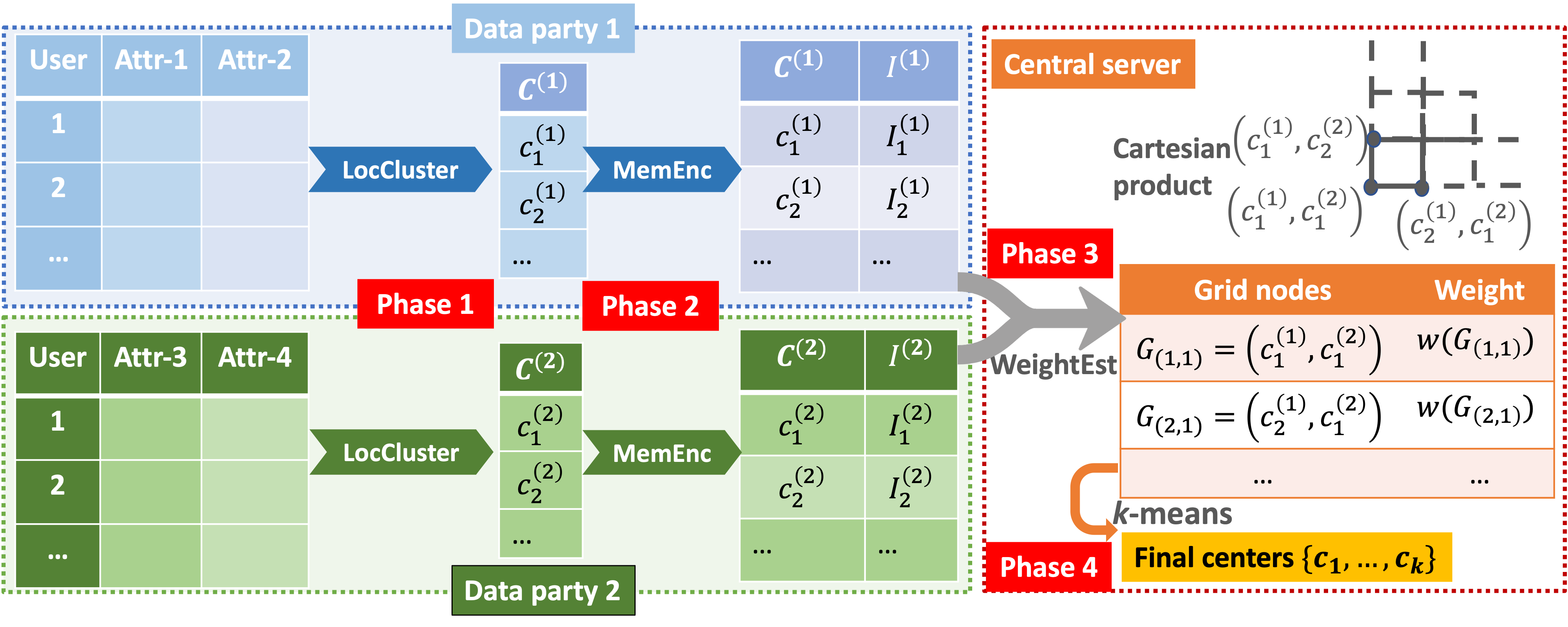}
        \vspace{-0.4cm}
        \caption{General framework for VFL \kmeans clustering.}
        \label{fig:vfl-kmeans}
    \end{minipage}
    \hfill
    \begin{minipage}{0.35\textwidth}
        \centering
        \includegraphics[width=\textwidth]{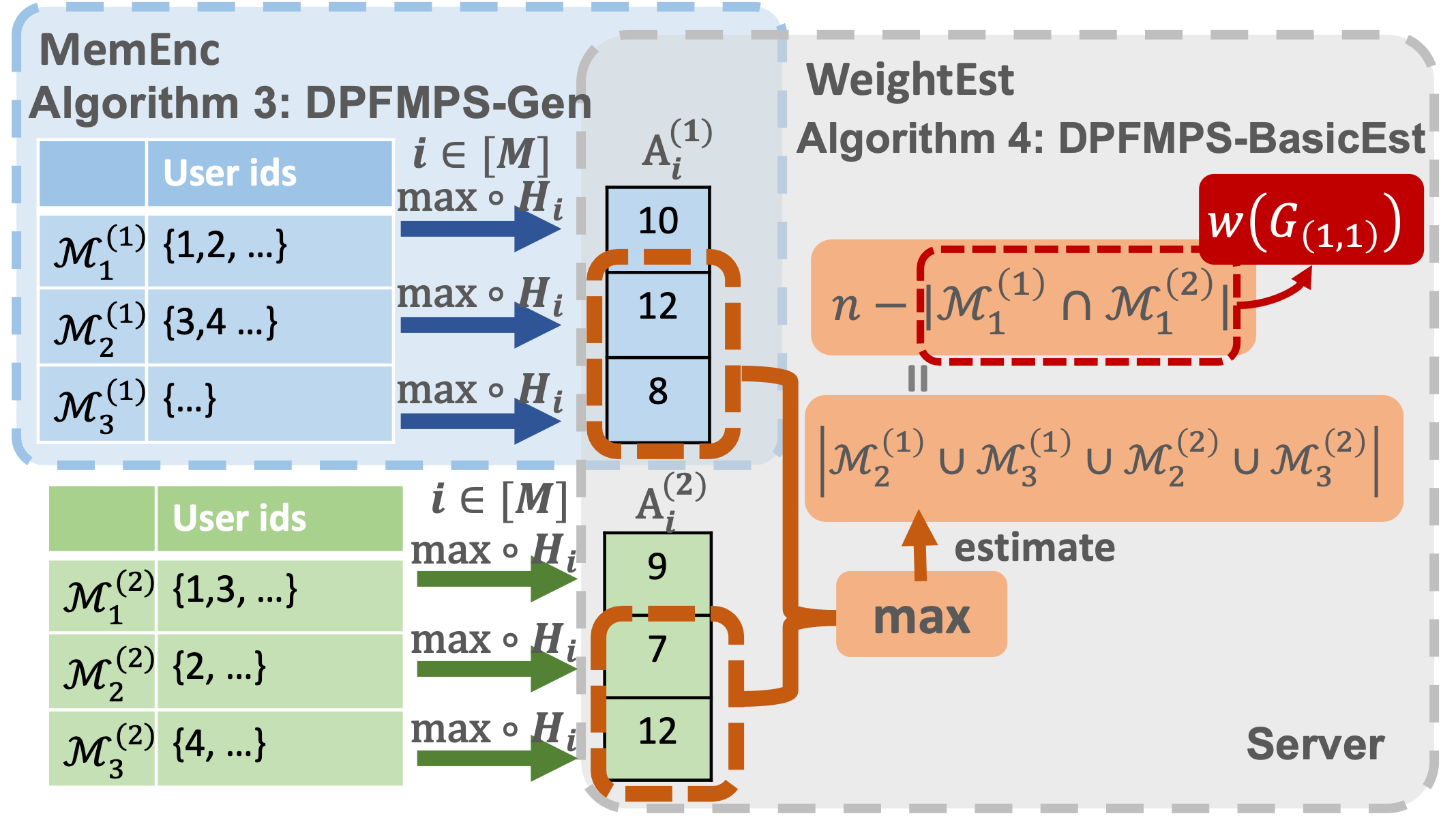}
        \vspace{-0.4cm}
        \caption{Example of Algorithm~\ref{algo:dpfmps} and \ref{algo:dpfms-est}.}
        \label{fig:algo34}
    \end{minipage}
    \hfill
\end{figure*}

\mypara{DPLSF.}
There are two recently proposed algorithms for differentially private \kmeans with theoretical performance guarantees, one for the central setting \cite{ghazi2020dpkmean} and one for the local setting \cite{chang2021locally}.  
Both algorithms are built on a theoretical concept called efficiently decodable net.  But how to implement the efficiently decodable net in practice is still unclear.
Therefore, the authors also propose a DP \kmeans algorithm based on \emph{locality sensitive hashing (LSH) forest}~\cite{charikar2002simhash} to approximate the effect of the efficiently decodable net.
The central DP implementation is open-sourced~\cite{googledpclustering}.
The high-level idea is to partition the data points based on their LSH outputs, generate differentially private means and counts for these partitions, and finally run a (non-private) \kmeans algorithm on the means with counts as weights.
We call this method DPLSF.
We choose DPLSF as the instantiation of \localclustering in this paper because it is shown to outperform other existing methods in experiments~\cite{googledpclustering}.

\mypara{Adapting DPLSF to VFL setting.} 
The implementation in \cite{googledpclustering} requires a known $L_2$ norm upper bound for the data points because of the usage of the Gaussian mechanism.
However, assuming the $L_2$ norm upper bound for each user's data may be unreasonable in the VFL setting because a user's data are spread in different data parties' datasets.
Thus, we normalize each attribute to some ranges to restrict the sensitivity of the data points averaging operations of DPLSF, and apply the Laplace mechanism to provide DP guarantees. 
Note that normalizing different attributes to different ranges is essentially adjusting the weights of different attributes when computing the distances. 
To simplify the discussion and experiment settings, we let each data party normalize its attributes to $[-1, 1]$.
However, our technique can be easily extended when different attributes are normalized to different target ranges, so long as these target ranges are public information, e.g., general domain knowledge;
otherwise, normalization ranges can be inferred using other DP algorithms with reserved private budgets.

\section{Private Membership Encoding and Weight Estimate }
\label{sec:membership}
To avoid the privacy leakage when sharing the membership information $\mathbf{I}\rel{}{\ell}$,
we introduce our private instantiations of \membershipencoding and \weightestimate together in this section because how the central server can estimate the weights with \weightestimate depends on how data parties encode the membership information with \membershipencoding.
With our instantiations, the data parties generate differentially private membership information $\mathbf{I}\rel{}{\ell}$ and share them with the central server, and the central server estimates the cardinalities of the intersections $
\left|\members\rel{a_1}{1} \bigcap \ldots \bigcap \members\rel{a_S}{S}\right|$ for all $(a_1, \ldots, a_S) \in [k']^S$ as weights.

\subsection{Baselines }
\label{subsec:baselines}

\mypara{Baseline 1: Estimate weights assuming independence among attributes.}
In this approach, we assume that the distributions of attributes from one party are independent of those from all other parties. 
Under this assumption, we can compute the intersection cardinality using $\left|\members\rel{a_1}{1} \bigcap \ldots \bigcap \members\rel{a_S}{S}\right| \approx \hat{n}\prod_{\ell \in [S]}\frac{|\members\rel{a_\ell}{\ell}|}{\hat{n}}$.

Following this idea, the private \membershipencoding only needs to generate a histogram of $\left[\left|\members\rel{1}{\ell}\right|, \ldots, \left|\members\rel{k'}{\ell}\right| \right]$ with Laplace mechanism and privacy budget $\epsilon_2$.
Denote the randomized histogram vector as $\tilde{\mathbf{f}}\rel{}{\ell}$.
The central server's sub-procedure \weightestimate is $w(\grid_{(a_1, \ldots, a_S)}) =  \hat{n}\prod_{\ell \in [S]}\frac{\tilde{\mathbf{f}}\rel{a_\ell}{\ell}}{\hat{n}}$.
We call this baseline as \indca because it makes the independence assumption and uses the Laplace mechanism.

However, when the assumption of inter-party attributes independence fails, this cardinality estimation can be far from the ground truth and make the final centers far from optimal, because the correlation information between the inter-party attributes is completely lost.
Thus, maintaining the inter-party attribute correlations is the main focus of improving the utility in general scenarios.

\mypara{Baseline 2: Estimate weights based on local differential privacy protocols.}
Another choice for aggregating the cardinality information is to let the parties report each user's membership information separately, instead of aggregating the membership information first and then reporting.  
When such reporting satisfies local differential privacy (LDP) for each user, it also satisfies DP for the whole local dataset. 
In this paper, we apply either the optimized local hashing (OLH) or the general random response (GRR) protocol in \cite{WangBLJ17} (which is used depends on the privacy parameter $\epsilon_2$ and the domain size), and name this approach as \ldpca.
We set $\epsilon_0=0$ because \ldpca does not need to estimate the number of users.

Local memberships of a user in $S$ different data parties can be seen as an $S$-dimension record.
Thus, it is equivalent to randomizing each ``dimension'' of a user record independently with LDP protocols.
After receiving all the local memberships of a user, the server first computes the probability vector of this user in all possible intersections $\members\rel{a_1}{1}\cap\ldots\cap\members\rel{a_S}{S}$.
Then the server sums the probability vectors of all users to get the desired weights.
The correlation of each user's attributes is preserved because the server first aggregates all local memberships of each user.  
We defer more details of \ldpca in the appendix of our full version~\cite{manuscript}.

However, the reported memberships are very noisy. 
Based on the known LDP protocol error analysis~\cite[Proposition 10]{wang2019answering}, the variance of an estimated cardinality with \ldpca is in the order $O\left(n/ \epsilon_2^{2S}\right)$ for each intersection.
The noise can easily overwhelm the true counts when $\epsilon_2$ is small or $S$ is large.


\subsection{Prerequisite: DP FM Sketch}
\label{subsec:psi}
As is shown, neither baseline is satisfactory.  
An effective approach for privacy-preserved \membershipencoding and \weightestimate should accurately maintain most of the inter-party correlation information.
We propose a new approach based on the Flajolet-Martin (FM) sketch because it can satisfy DP with a little additional overhead and support the \emph{set union operation}.

\mypara{FM sketch achieves DP.}
As mentioned in Section~\ref{sec:background}, FM sketches are used to estimate cardinality.  Recently, some research results show that a family of sketches, including FM sketch, satisfy DP as long as the cardinality is large enough and the hash keys are unknown to the adversary~\cite{smith2020fmsketch, hu2021ca, choi2020dpsketch, dickens2022allsketch}.
The DP version of the FM sketch algorithm is described as Algorithm~\ref{alg:fmsketch} following the approach in Smith et al.~\cite{smith2020fmsketch}, where the FM sketch is implemented using
an ideal geometric-value hash function $H: \mathcal{X}\times \mathbb{Z} \rightarrow \mathbb{N}_+$ with parameter $\frac{\gamma}{1 + \gamma}$. 
That is, given any finite set of distinct inputs $x_1,\ldots, x_\ell \in \mathcal{X}$, 
with a hash key $\zeta \sim \text{Uniform}(\mathbb{Z})$, the hashed values $H_\zeta(x_1), \ldots, H_\zeta(x_\ell)$ follow i.i.d. Geometric$\left(\frac{\gamma}{1+\gamma}\right)$ distribution.

\begin{algorithm}
\caption{DP FM Sketch Generation $\dpfm$ \cite{smith2020fmsketch}}
\label{alg:fmsketch}
\begin{algorithmic}[1]
\Require a (multi)set $\members$, privacy parameter $\epsilon'$, and Geometric distribution parameter $\gamma$, an ideal random hash function $H_{\zeta}(x) \sim \geometric(\frac{\gamma}{1 + \gamma})$ when $\zeta \sim \uniform(\mathbb{Z})$
\Ensure Sketch $\alpha$ for cardinality of $\members$
\State $n_p = \lceil\frac{1}{e^{\epsilon'} - 1}\rceil$, $\alpha_{\min} = \lceil \log_{1+\gamma}\frac{1}{1 - e^{-\epsilon'}} \rceil$
\State $\alpha_p = \max\{Y_1, \ldots, Y_{n_p}\}$ where $Y_i \sim \geometric(\frac{\gamma}{1+\gamma})$ 

\State $\alpha_{real} = \max\{H_\zeta(\id) | \id \in \members\}$ 
\State Return $ \alpha = \max\{\alpha_p, \alpha_{real}, \alpha_{\min}\}$
\end{algorithmic}
\end{algorithm}

Algorithm~\ref{alg:fmsketch} generates a DP FM sketch. 
The intuition of the non-private FM sketch is that when elements in $\members$ are encoded as a set of geometric random variables and $\alpha_{real} = \max\{H_\zeta(\id) | \id \in \members\}$, we can expect $(1+\gamma)^{\alpha_{real}} \in \left[\frac{|\members|}{1+\gamma}, (1+\gamma)\cdot|\members| \right]$ with reasonable probability.
Compared with the non-private version, the DP FM sketch needs two additional steps to ensure privacy: adding phantom elements, and lower-bounding the output by $\alpha_{\min}$.
The phantom elements are used to ensure that the cardinality estimated by the final output is at least $n_p$; the $\alpha_{\min}$, which is at the $e^{-\epsilon'}-$quantile of the Geometric distribution, is used to ensure a probability that none of the items affects the output.
It has been shown that the harmonic/geometric mean of $M$ runs of Algorithm~\ref{alg:fmsketch} satisfies DP:

\begin{lemma}[Privacy guarantee of $\dpfm$ \cite{smith2020fmsketch}]
\label{lemma:fmpriv}
    Given an ideal geometric-value hash function $H$, Algorithm~\ref{alg:fmsketch} is $\epsilon'$-DP.
    Besides, repeating it $M$ times with different hash keys and $\epsilon'=\frac{\epsilon}{4\sqrt{M \log(1/\delta)}}$ satisfies $(\epsilon, \delta)$-DP as long as $\epsilon\leq 2\log(1/\delta)$.
\end{lemma}

\subsection{Sketch-based \membershipencoding and \weightestimate}
Unlike the DP FM sketch~\cite{smith2020fmsketch} which aims to estimate the cardinality of a single set, our task is to encode the partitions, namely multiple local clusters consisting of user ids assigned to different local centers.
Thus, we extend the FM sketch to encode the partition memberships, called \underline{D}ifferentially \underline{P}rivate \underline{F}lajolet-\underline{M}artin \underline{P}artition \underline{S}ketch (\dpfmps).
We call an FM sketch vector for $\left\{\members\rel{1}{\ell},\ldots, \members\rel{k'}{\ell}\right\}$ as a set of FM sketches.
The sketch generation function, \dpfmpsgen (Algorithm~\ref{algo:dpfmps}), is an instantiation of \membershipencoding. 
The data parties need to share a set of auxiliary parameters, $\aux=\left\{\gamma, M, \vzeta=\left\{\zeta_1, \ldots, \zeta_M \right\} \right\}$, where $\gamma$ is the Geometric distribution parameter, $M$ is the number of sets of sketches and each $\zeta_i\sim\uniform(\mathbb{Z})$ is the hash key to an ideal geometric-value hash function for the $i$-th set of sketch.
Notice that all data parties need to share the same set of hash keys.
Since the hash keys must be unknown to the central server to achieve DP with the FM sketches, each data party can generate a random number and share it with all other parties via some secure peer-to-peer channels (e.g., key-exchange protocol~\cite{diffie1976new}); then, each data party can use the sum/XOR/concatenation of those $S$ random numbers as a hash key.
Algorithm~\ref{algo:dpfmps} generates $M$ sets of FM partition sketches and has the following privacy guarantee.

\begin{theorem}\label{thm:fmpriv}
    Given an ideal geometric-value hash function $H$, \dpfmpsgen (Algorithm~\ref{algo:dpfmps}) generating $M$ sets of partition sketches satisfies $(\epsilon_2, \delta_2)$-differential privacy.
\end{theorem}

The above theorem is based on the following lemma about the privacy guarantee for each set of sketch, namely $\FM_i$ in Algorithm~\ref{algo:dpfmps}.

\begin{lemma}
\label{lemma:single}
    Given an ideal geometric-value hash function $H$, each set of the partition sketch (i.e. a row $\FM_i$) is $\epsilon'$-DP.
\end{lemma}

With Lemma~\ref{lemma:single}, the privacy guarantee claimed in Theorem~\ref{thm:fmpriv} can be derived with the sequential composition of RDP~\cite{mironov2017renyi} and converted back to $(\epsilon, \delta)$-DP  following the same proof as in~\cite{smith2020fmsketch}.

\begin{algorithm}
\caption{\dpfmpsgen} \label{algo:dpfmps}
\begin{algorithmic}[1]
\Require  A set of $k'$ centers $\centroids\rel{}{\ell}$, dataset $\X\rel{}{\ell}$, privacy parameter $\epsilon_2$ and $\delta_2$,
Geometric distribution parameter $\gamma$, number of sketches $M$ and the corresponding hash keys $\vzeta=\{\zeta_1, \ldots, \zeta_M\}$
\Ensure $M$ sets of sketches $\FM$
\State $\epsilon' = \frac{\epsilon_2}{4\sqrt{M \log (1/\delta_2)}}$
\State Generate $\left\{\members\rel{1}{\ell},\ldots, \members\rel{k'}{\ell} \right\}$ where $\id \in \members\rel{j}{\ell}$ based on $\centroids\rel{}{\ell}$
\State $\FM \leftarrow \mathbf{0}^{M\times k'}$
\For{$i \in [M], a \in [k']$ }
    \State $\FM_{i, a} = \dpfm(\members\rel{a}{\ell}, \epsilon', \gamma, \zeta_i)$
\EndFor
\State Return $\FM$ 
\end{algorithmic}
\end{algorithm}

\begin{algorithm}
\caption{\dpfmpsest }
\label{algo:dpfms-est}
\begin{algorithmic}[1]
\Require Estimate of the total number of users $\hat{n}$, Privacy parameter $\epsilon_2$ and $\delta_2$, the FM partition sketches $\{\FM\rel{}{1}, \ldots, \FM\rel{}{s}\}$ 

\Ensure Estimate of the weights $w(\grid)$
\State $\mathbf{U} \leftarrow \mathbf{0}^{M \times (k')^s}$, $\cca \leftarrow \mathbf{0}^{(k')^s} $ \label{algstep:init}
\State $\epsilon' = \frac{ \epsilon_2}{4\sqrt{M \log (1/\delta_2)}}$,  $n_p = \lceil\frac{1}{e^{\epsilon'} - 1}\rceil$ 
\For{$i \in [M]$, $(a_1, \ldots, a_s) \in [k']^{s}$}
    \State $ \mathbf{U}_{i, (a_1, \ldots, a_s)} = \max\{\FM\rel{i, a}{\ell} \mid \ell \in [s], a \neq a_\ell \}$ \label{algstep:unionfm}
\EndFor
\For{$(a_1, \ldots, a_s) \in [k']^s$}
    \State $\fm_{(a_1, \ldots, a_s)} = \harmonic{ U_{:, (a_1, \ldots, a_s)} }$ \label{algstep:union_est}
    \State $\cca_{(a_1, \ldots, a_s)} = (1+\gamma) ^ {\fm_{(a_1, \ldots, a_s)}} - s (k'-1) n_p$ \label{algstep:union_debias}
    \State $w(\grid_{(a_1, \ldots, a_s)}) = \hat{n} - \cca_{(a_1, \ldots, a_s)}$  \label{algstep:intersect_est}
\EndFor
\State Ensure $\sum_{(a_1, \ldots, a_s)}w(\grid_{(a_1, \ldots, a_s)}) = \hat{n} $ and $w(\grid)_{(a_1, \ldots, a_s)} \geq 0$ \label{algstep:consistent}
\State Return $w(\grid)$
\end{algorithmic}
\end{algorithm}

\mypara{Estimate intersection cardinality.}
One key observation of the local partition is that an $\id$ can be clustered to one and only one $\members\rel{a}{\ell}$ by each data party $\ell$.
Thus, the intersection problem can be transformed into the union problem by an extension of the inclusion-exclusion principle:
\begin{align}
\small
    \bigcap_{\ell=1}^{S} \members\rel{a_\ell}{\ell} = \overline{\bigcup_{\ell=1}^{S} \overline{\members\rel{a_\ell}{\ell}} } = \overline{\bigcup_{\ell=1}^{S} \left(\bigcup_{a \neq a_\ell} \members\rel{a}{\ell} \right) } . \label{eq:inclusion_exclusion}
\end{align}

Algorithm~\ref{algo:dpfms-est} gives the detailed procedure.  There are $(k')^S$ possible intersections that need to be estimated, and there are $M$ sets of FM sketches.
In Line~\ref{algstep:init}, $\mathbf{U}$ is initialized to store the FM sketches for the cardinalities of the complementary set of intersections; $\cca$ is an intermediate vector representing the cardinalities of union of complements, i.e. 
$\left| \bigcup_{\ell=1}^{s}\overline{\members\rel{a_\ell}{\ell}} \right|$.

To estimate the cardinalities of the intersection with Equation~\eqref{eq:inclusion_exclusion}, we first need to calculate the sketch of their complementary set, i.e., $\overline{\members\rel{a_\ell}{\ell}}$. 
The corresponding FM sketch can be obtained by taking the max of all other elements in the same sketch set, $\max\left\{A\rel{i, a}{\ell} | a\neq a_\ell \right\}$, according to the union property of FM sketch.
Next, we need to derive the sketch for the union of the complementary partition of all data parties, i.e. $\bigcup_{\ell=1}^{s}\overline{\members\rel{a_\ell}{\ell}}$.
This union's FM sketches can be obtain by $ \max\left\{\max\left\{A\rel{i, a}{\ell} | a\neq a_\ell\right\} | \ell\in [s] \right\}$.
Merging the two max operations give us the operation in Line~\ref{algstep:unionfm}.

We use the Harmonic mean over the $M$ FM sketches to estimate the cardinality of $\bigcup_{\ell=1}^{s} \overline{\members\rel{a_\ell}{\ell}}$ in Line~\ref{algstep:union_est},
because it is shown~\cite{smith2020fmsketch} that the harmonic mean estimate is more stable and accurate.
As this sketch is obtained after $s(k'-1)$ union operations, totally $s(k'-1)n_p$ phantom elements are taken into account.
Therefore, we need to subtract $s(k'-1)n_p$ elements from the final estimation in Line~\ref{algstep:union_debias}.
Finally, the intersection cardinality can be calculated by subtracting the cardinalities of $\bigcup_{\ell=1}^{s} \overline{\members\rel{a_\ell}{\ell}}$ from the total number of users $\hat{n}$.
Before finally returning the weights, we need to make sure the output is valid by enforcing non-negativity and total-sum equal to $\hat{n}$ on the weights (Line~\ref{algstep:consistent}).
An example of running Algorithm~\ref{algo:dpfmps} and \ref{algo:dpfms-est} is shown in Figure~\ref{fig:algo34} with two data parties and $k'=3$.

\subsection{Privacy, Utility and Communication Cost}
\label{subsec:utility}
The proofs of the theorems in this subsection are provided in the appendix of the full version~\cite{manuscript} because of the space limitation.

\mypara{Privacy guarantee.}
According to the privacy spitting strategy in Algorithm~\ref{algo:vfc}, Theorem~\ref{thm:fmpriv} and the sequential composition of DP, the following privacy statement can directly derived to describe the privacy loss from all data parties to the central server, and equivalently, the final output of the central server.
\begin{theorem}
\label{thm:final-privacy}
    Algorithm~\ref{algo:vfc} is $(\epsilon, \delta)$-DP with \dpfmpsgen and \dpfmpsest as the implementation of \membershipencoding and \weightestimate.
\end{theorem}

\mypara{Error analysis of the weights.}
The cardinality estimate with Algorithm~\ref{alg:fmsketch} approximates the real cardinality within a factor of $1 \pm \gamma$ and an additive error of $O\left(\sqrt{\ln (1/\delta)}/\epsilon \right)$~\cite{smith2020fmsketch}.
The additive error is because of the phantom elements, and $\alpha_{\min}$ also slightly increases the expectation of the sketch. 
We provide a refined result to show that the utility of the private FM sketches generated from Algorithm~\ref{alg:fmsketch} is similar to the non-private FM sketch when the real cardinality is large enough, in order to use the more advanced results from non-private FM sketch research.

\begin{lemma}
\label{lemma:expectation+variance}
    Let $\hat{\alpha} = \max\{\alpha_{real}, \alpha_p\}$ and $\alpha$ be the return of Algorithm~\ref{alg:fmsketch}.
    If $H$ is an ideal geometric-value hash function and $\epsilon'$ is fixed, we have $\EV{\alpha} / \EV{\hat{\alpha}} \rightarrow 1$ and $\Var{\alpha} /  \Var{\hat{\alpha}}\rightarrow 1$ when $|\members|$ is sufficiently large.
\end{lemma}

With Lemma~\ref{lemma:expectation+variance}, we can claim that the lower bound $\alpha_{\min}$ does not affect the mean and variance of the sketch $\hat{\alpha}$ too much.
So we can treat the $\alpha$ returned by Algorithm~\ref{alg:fmsketch} approximately the same as the vanilla non-private FM sketch,
and we can use the standard deviation results of the non-private FM sketch to analyze the error of \dpfm in a more fine-grained way.

The standard deviation of a non-private FM sketch estimation can be represented as $\rho N / \sqrt{M}$, where $\rho$ is a constant, $N$ is the cardinality and $M$ is the repetition~\cite{flajolet1985,lang2017back,flajolet2007hyperloglog}.
Our algorithm has $s(k'-1)$ union operations to derive an element of $\mathbf{U}$, and each set has $n_p$ phantom elements.
So the cardinality estimated by $\mathbf{U}_{:, (a_1, \ldots, a_S)}$ becomes $N_{(a_1, \ldots, a_S)}=\left|\bigcup_{\ell \in [S]}\cup_{j\neq a_\ell} \members\rel{j}{\ell}\right| + s(k'-1)n_p = \hat{n} - w^*(\grid_{a_1, \ldots, a_S}) + s(k'-1)n_p$, where $ w^*(\grid_{a_1, \ldots, a_S})$ represent the true intersection cardinality $|\members\rel{a_1}{1} \cap \ldots \cap\members\rel{a_S}{S}|$.
Based on the property of the non-private FM sketch and the value of $M$ and $n_p$, we can state the following lemma.

\begin{theorem}
\label{thm:sketch-std}
    Given a constant $\rho$ such that the non-private FM sketch's standard deviation is $\rho N / \sqrt{M}$ where $N$ is the cardinality, and $M$ is number of repetitions. 
    With $n_p$ and $\epsilon'$ set as in Algorithm~\ref{algo:dpfms-est}, each intersection cardinality estimate generated has standard deviation \begin{align*}
        \sigma_{(a_1, \ldots, a_s)} = \frac{\rho (n - w^*(\grid_{(a_1, \ldots, a_s)}))}{\sqrt{M}} +  \frac{4\rho s(k'-1)\sqrt{\log (1/\delta)}}{\epsilon_2}
    \end{align*} 
\end{theorem}

The result is directly derived after plugging in the value of $N_{(a_1, \ldots, a_S)}$ and $n_p$, and use the approximation $e^{y}\approx y+1$ when $y$ is a small positive number.

\mypara{Final utility cost.}
When we set $k'=k$ as indicated in the non-private setting~\cite{ding2016k}, we can show that Algorithm~\ref{algo:vfc} is a $(\beta, \lambda)$-approximation algorithm with assumption of accessibility to some guaranteed central (private) \kmeans algorithm.

\begin{theorem}
\label{thm:final-utility}
    Assume the data parties have access to a differentially private ($\beta_{\mathit{priv}}, \lambda_{\mathit{priv}}$)-approximate \kmeans algorithm, and the central server has access to a (non-private) ($\beta_{o}, 0$)-approximate \kmeans algorithm. 
    Algorithm~\ref{algo:vfc} with \dpfmpsgen and \dpfmpsest is a $(\beta, \lambda)$-approximation algorithm with probability $1-\omega$, where
    \begin{align*}
        & \beta =  2\beta_{priv} + 4\beta_0 + 4 \beta_0 \beta_{\mathit{priv}} , \\
        & \lambda =  2(\beta_0 + 1)S\lambda_{\mathit{priv}} 
        + O\left( \frac{\beta_0 m^2 k^{1.5S}}{\sqrt{\omega}} \left(\frac{n}{\sqrt{ M}} +  \frac{S(k-1)\sqrt{\log (1/\delta)}}{\epsilon_2}  \right) \right)
    \end{align*}
\end{theorem}

The multiplicative error $\beta$ is composed of $\beta_{\mathit{priv}}$ from the \localclustering and $\beta_0$ from the final non-private clustering on central server.
By the latest theoretical result~\cite{ghazi2020dpkmean}, $\beta_{priv}$ and $\beta_0$ are close to 1 when $k$ is a constant. 
So our approximate ratio ($\beta \approx 10$) is slightly larger than the bound in the non-private algorithm ($\beta\approx 9$)~\cite{ding2016k}, because of the randomness when estimating the cardinality satisfying DP.
Besides, DP \kmeans algorithms are unavoidable to have an additive error besides the multiplicative error~\cite{ghazi2020dpkmean}. 
The first term of our additive error $\lambda$ can be understood as the cumulative error of $S$ local DP \kmeans algorithm results\footnote{According to the best theoretical results of DP \kmeans in central setting \cite{ghazi2020dpkmean}, $\lambda_{\mathit{priv}}= O_{\beta, \eta}\left(\epsilon_1^{-1}(km + k^{O_{\eta}(1)})\polylog n \right)$ with a small positive constant $\eta$.}, and the second term comes from the cumulative cardinality estimation error of $k^S$ nodes.
The second term can dominate the error when the number of centers $k$ or the number of data parties $S$ is not small.
If both $n$ and $M$ are large enough, then the averaged additive error (divided by $n$) will still vanish.
Besides, we also empirically show that the losses can be small and even close to the central private \kmeans losses on some datasets when the privacy budget is large enough.

\mypara{Communication and computation cost.}
There is only one round of communication between the data parties and the central server.
Compared with the non-private baseline~\cite{ding2016k}, the additional computation cost for \dpfmpsgen on each data party is $O(nM)$ hashing operations for $M$ sets of sketches.
However, this process can be easily accelerated by parallel computation because each \dpfm can run independently.
The communication cost is $O\left((m\rel{}{\ell}+M)k' \right)$. 
Our algorithm's communication cost is independent of $n$ and it can even smaller than the non-private solution $O\left(m\rel{}{\ell}k' + n\right)$~\cite{ding2016k} when $n > Mk'$.
Compared with the non-private algorithm requiring $O(nS)$ operations for the intersection cardinality, our private algorithm needs $O(MSk'^S)$ for the server to estimate the weights.

\section{Improving Utility of the Algorithm}
\label{sec:improve}
We introduce two heuristic methods in this section to improve the empirical performance of our methods when $S$ is large.

\subsection{Improving Post-processing Estimation Algorithm for More than Two Data Parties}
\label{subsec:2way}

As the number of data parties increases, the accuracy of the weight estimation will decrease dramatically.  
From the error analysis of the cardinality estimation (Theorem~\ref{thm:sketch-std}), one can see that the second term in standard deviation $\sigma_{(a_1, \ldots, a_S)}$ increases proportionally with the number of data parties $S$.  
Furthermore, the total possible intersection combinations grow exponentially as $(k')^S$, which means fewer expected number of data points fall in the intersection, i.e., the expected $w^*(\grid_{(a_1, \ldots, a_S)})$.
The combination of the two factors means that the relative error of the intersection estimation explodes as the number of parties grows.  

Two observations give us hope to lessen the negative effect.
The first observation is that estimation of two-party intersection cardinalities is relatively accurate.
The second observation is based on the distributive property of set intersection:
\begin{align}
\small
    \left|\members\rel{a'_{\ell_1}}{\ell_1} \bigcap \members\rel{a'_{\ell_2}}{\ell_2} \right| = \sum_{\substack{(a_1, \ldots, a_S)\in [k']^S \\ a_{\ell_1} = a'_{\ell_1}, a_{\ell_2} = a'_{\ell_2}}} \left| \bigcap \members\rel{a_\ell}{\ell} \right| ,
    \label{eq:sum_intersection}
\end{align}
which give the connection between the all-party intersection cardinality (i.e., $w(\grid)$) and the two-party version (i.e., $w(\grid\rel{}{\ell_1, \ell_2})$). 
Thus, we propose to deal with this challenge by 1) {\it computing only all pair-wise intersection cardinalities}, and 2) using these two-party intersection cardinalities with Equation~\eqref{eq:sum_intersection} as constraints, and iteratively update the $w(\grid)$ to fulfill all these constraints.

The improved estimation is described in Algorithm~\ref{algo:multiparty} \dpfmpsimprove.
The central server first estimates the single-party cardinality of all local clusters with only $\FM\rel{}{\ell}$ (Line~\ref{algstep:estimate_one}).
Namely, $w\rel{a}{\ell}$ is an estimate for $|\members\rel{a}{\ell}|$.
Then the server initializes the full grid weights $w(\grid)$ with only the single-party cardinality information assuming no correlation between the inter-party attributes  (Line~\ref{algstep:init_uniform}).
Next, the server estimates all two-party weights with \dpfmpsest as a sub-procedure (Line~\ref{algstep:pair_weights}). 
To be more detailed, the grid weights $w(\grid\rel{}{\ell_1,\ell_2})$ returned by \dpfmpsest with $\FM\rel{}{\ell_1}$ and $\FM\rel{}{\ell_2}$ are the estimates for  $\left\{\left|\members\rel{a_{\ell_1}}{\ell_1}\bigcap \members\rel{a_{\ell_2}}{\ell_2} \right| \ \mid \ a_{\ell_1} , a_{\ell_2} \in [k']\right\}$.
With all pairs of $w(\grid\rel{}{\ell_1, \ell_2})$, the server iteratively updates the grid weights $w(\grid)$ to make them consistent with all the two-party intersection cardinalities $w(\grid\rel{}{\ell_1, \ell_2})$.
In each iteration, the server first randomly selects a pair of data parties (Line~\ref{algstep:pair_party}), groups the current full grid weights $w(\grid)$ by the cluster indices of the chosen two parties, and sums the weights in the same group to generate a two-party grid weights $\tilde{w}(\grid\rel{}{\ell_1, \ell_2})$ according to Equation~\eqref{eq:sum_intersection}:
{\small
$$
    \tilde{w}(\grid\rel{(a'_{\ell_1}, a'_{\ell_2})}{\ell_1, \ell_2}) =  \sum_{\substack{(a_1, \ldots, a_S)\in [k']^S \\ a_{\ell_1} = a'_{\ell_1}, a_{\ell_2} = a'_{\ell_2}}}w(\grid_{(a_1, \ldots, a_S)})
$$
}
We use the differences between $w(\grid\rel{}{\ell_1, \ell_2})$ and $\tilde{w}(\grid\rel{}{\ell_1, \ell_2})$ to update the weight evenly with a step size $\eta_t$ (Line~\ref{algstep:update}).
After sufficient update iterations, the full-party grid weight $w(\grid)$ is expected to  approximately satisfy the constraints (Equation~\eqref{eq:sum_intersection}) with all two-party weights $w(\grid\rel{}{\ell_1, \ell_2})$.
Because both \dpfmpsimprove and \dpfmpsest are post-processing components in the DP definition, the privacy guarantee in Theorem~\ref{thm:final-privacy} still holds for \dpfmpsimprove.

\begin{algorithm}
\caption{\dpfmpsimprove} \label{algo:multiparty}
\begin{algorithmic}[1]
\Require Estimated total number of users $\hat{n}$, privacy parameter $\epsilon_2$ and $\delta_2$, the FM partition sketches $\{\FM\rel{}{1}, \ldots, \FM\rel{}{s}\}$ 
\Ensure  Estimate of the weights $w(\grid)$
\State  $\epsilon' = \frac{\epsilon_2}{4\sqrt{M \log (1/\delta_2)}}$, $n_p = \lceil\frac{1}{e^{\epsilon'} - 1}\rceil$
\For{$\ell \in [S], a\in [k']$}
\State $w\rel{a}{\ell}\leftarrow (1+\gamma)^{\harmonic{\FM\rel{:, a}{\ell}}} - n_p$. \label{algstep:estimate_one} 
\EndFor
\State  $\forall \{a_1, \ldots, a_S\}\in[k']^S, w(\grid_{(a_1, \ldots, a_S)}) \leftarrow \hat{n} \times \prod_{\ell\in [S]}\frac{w\rel{a_\ell}{\ell}}{\hat{n}}$ \label{algstep:init_uniform} 

\For{$(\ell_1, \ell_2) \in [S]$ and $\ell_1 \neq \ell_2$}
    \State $w(\grid^{(\ell_1, \ell_2)})\leftarrow \dpfmpsest(\hat{n}, \epsilon_2, \delta_2, \{\FM\rel{}{\ell_1}, \FM\rel{}{\ell_2}\} )$\label{algstep:pair_weights}
\EndFor
\For{$t \in [T]$}  
    \State Randomly select a pair $(\ell_1, \ell_2) \in [S] \times [S]$ \label{algstep:pair_party}
    \State $\Tilde{w}(\grid^{(\ell_1, \ell_2)}) \leftarrow \texttt{Proj}(w(\grid), \ell_1, \ell_2)$ \label{algstep:proj}
    \State $\Delta\rel{}{\ell_1, \ell_2}\leftarrow \Tilde{w}(\grid\rel{}{\ell_1, \ell_2}) - w(\grid\rel{}{\ell_1, \ell_2})$ 
    \For{$(a'_{\ell_1}, a'_{\ell_2}) \in [k']^2$}
    \State $\forall (a_1, \ldots, a_S)\in [k']^S, a_{\ell_1} = a'_{\ell_1}, a_{\ell_2} = a'_{\ell_2}$
    \State $w(\grid_{(a_1, \ldots, a_S)}) \leftarrow w(\grid_{(a_1, \ldots, a_S)}) - \frac{\eta_t}{(k')^{S-2}}\Delta^{(\ell_1, \ell_2)}_{(a'_{\ell_1}, a'_{\ell_2})}$ \label{algstep:update} 
    \EndFor
\EndFor
\State Return $w(\grid)$
\end{algorithmic}
\end{algorithm}

\subsection{Auto-adjusted $k'$}
\label{subsec:vary-local-k}
The utility result of Theorem~\ref{thm:final-utility} is given by assuming the number of local and central clusters is the same, i.e., $k'=k$.
However, we can see a trade-off on the value of $k'$.
In the non-private setting, the larger the $k'$ is, the more local dataset information is preserved for the final central clustering.
On the other hand, the larger the $k'$ is, the more fine-grained the grid becomes, and the fewer records (or smaller $w^*(\grid_{a_1, \ldots, a_S})$) are expected to be assigned to the grid node.
According to Theorem~\ref{thm:sketch-std}, a larger $k'$ can make the error of $w(\grid_{a_1, \ldots, a_S})$ larger and increase the final cost.

Giving a closed form solution for the local $k'$ to minimize the final loss is difficult.
However, as we can see from the error bound of Theorem~\ref{thm:final-utility}, the grid quality reflected by the second term of $\lambda$ plays an important role.
We propose an empirical rule to set $k' = \max\{k_0, k^{1/S} \}$ for \dpfmpsimprove to prevent the true cardinalities from being overwhelmed by the noise, where $k_0$ is the smallest integer that satisfies $2\sigma_{(a_1, \ldots, a_s)} \geq \frac{\hat{n}}{k_0^2}$.
Because $w^*(\grid_{(a_1, \ldots, a_S)})$ in Theorem~\ref{thm:sketch-std} is unknown, we approximate it as $\frac{\hat{n}}{k_0^2}$; we set $s=2$ in the standard deviation because \dpfmpsimprove only decodes pairs of intersection cardinalities; and we set $\rho=0.649$ according to \cite{lang2017back}.
We experimentally demonstrate that such a choice of $k'$ can be good choices for the datasets in our experiments in Section~\ref{subsec:exp_vary_local_k}.

\section{Experiments}
\label{sec:experiments}

\begin{figure*}
    \centering
    \begin{subfigure}[b]{0.9\textwidth}
        \centering
        \includegraphics[width=\textwidth]{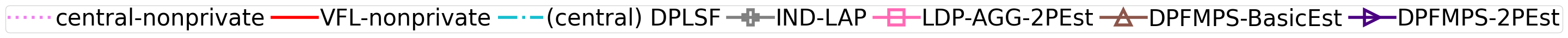}
    \end{subfigure}
    \\
     \begin{subfigure}[b]{0.24\textwidth}
         \centering
         \includegraphics[width=\textwidth]{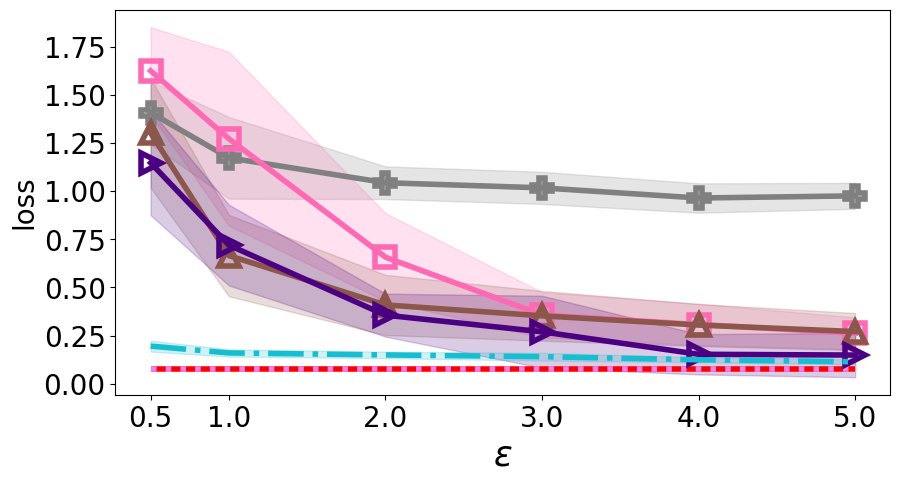}
         \vspace{-0.4cm}
         \caption{Mixed Guassian $S=2$}
         \label{fig:e2e_mg_2}
     \end{subfigure}
     \begin{subfigure}[b]{0.24\textwidth}
         \centering
         \includegraphics[width=\textwidth]{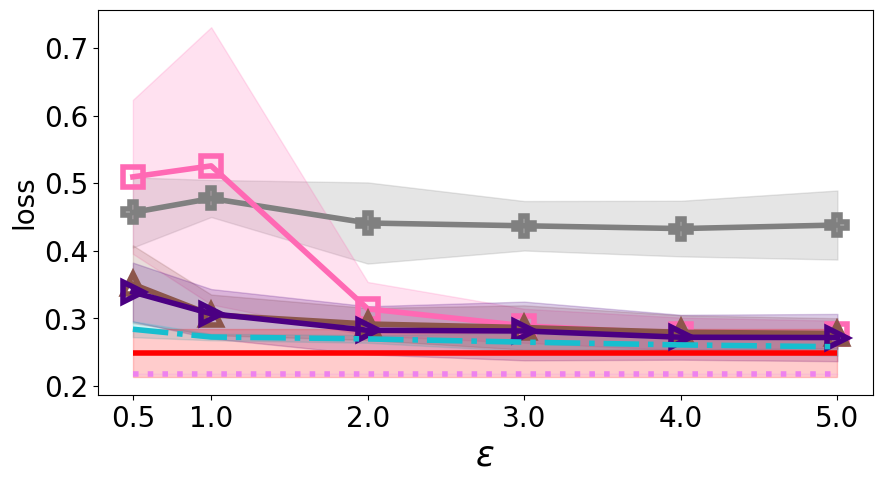}
         \vspace{-0.4cm}
         \caption{Taxi $S=2$}
         \label{fig:e2e_taxi_2}
     \end{subfigure}
     \begin{subfigure}[b]{0.24\textwidth}
         \centering
         \includegraphics[width=\textwidth]{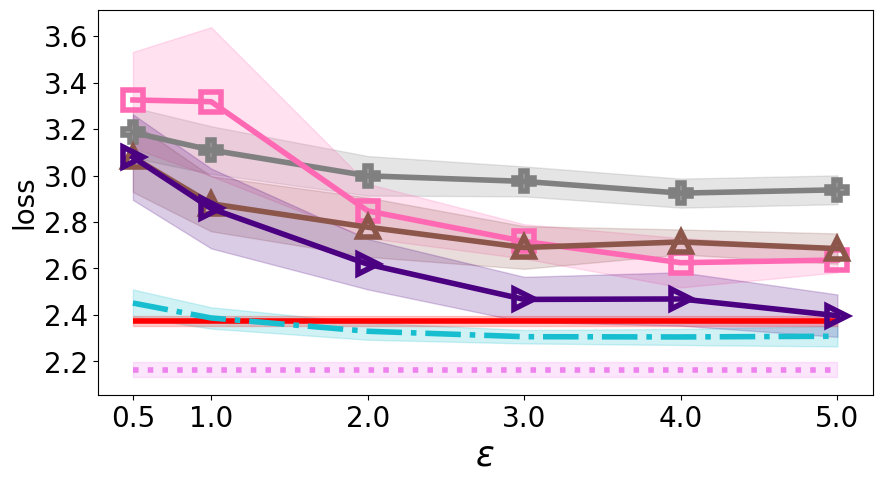}
         \vspace{-0.4cm}
         \caption{Loan $S=2$}
         \label{fig:e2e_loan_2}
     \end{subfigure}
     \begin{subfigure}[b]{0.24\textwidth}
         \centering
         \includegraphics[width=\textwidth]{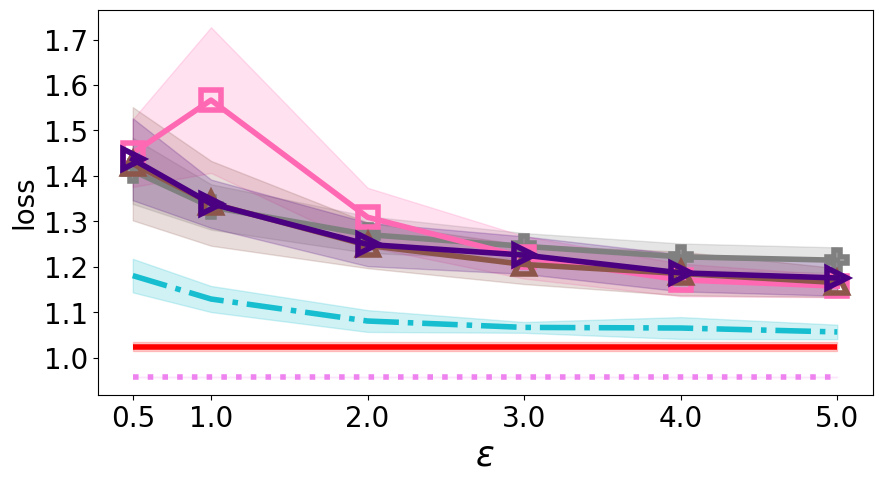}
         \vspace{-0.4cm}
         \caption{Letter $S=2$}
         \label{fig:e2e_letter_2}
     \end{subfigure}
     \\
     \begin{subfigure}[b]{0.24\textwidth}
         \centering
         \includegraphics[width=\textwidth]{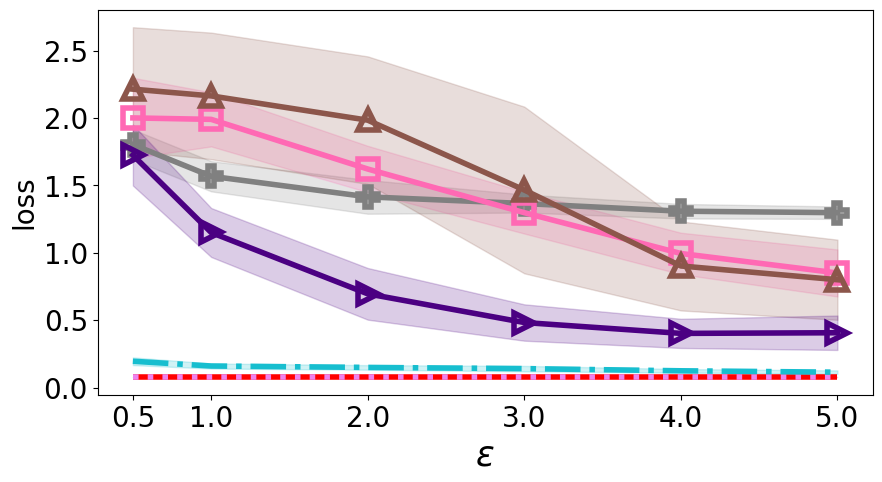}
         \vspace{-0.4cm}
         \caption{Mixed Guassian $S=4$}
     \end{subfigure}
     \begin{subfigure}[b]{0.24\textwidth}
         \centering
         \includegraphics[width=\textwidth]{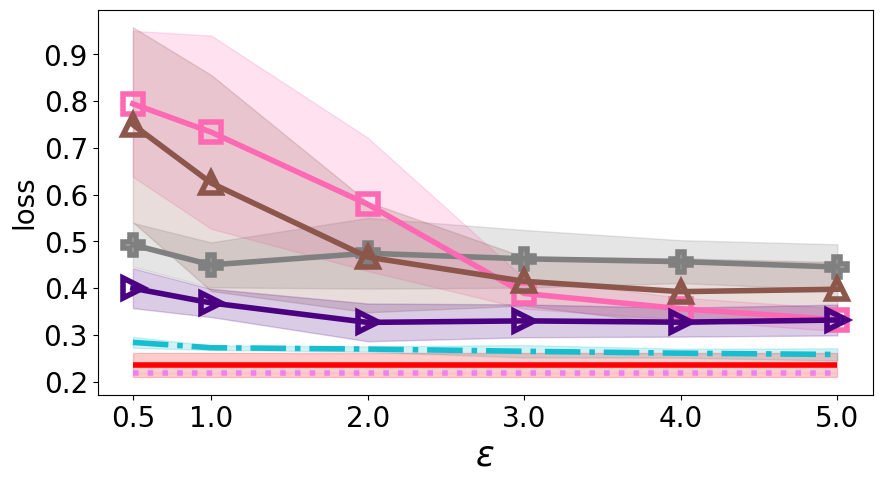}
         \vspace{-0.4cm}
         \caption{Taxi $S=4$}
         \label{fig:e2e_taxi_4}
     \end{subfigure}
     \begin{subfigure}[b]{0.24\textwidth}
         \centering
         \includegraphics[width=\textwidth]{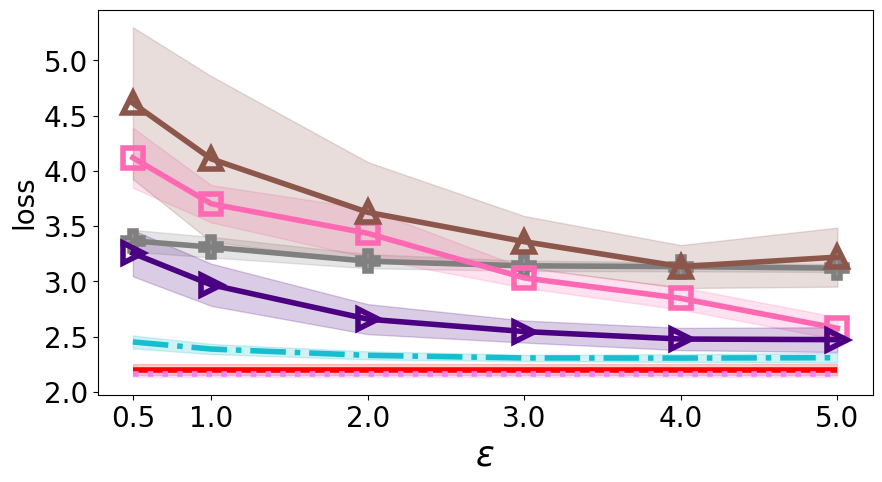}
         \vspace{-0.4cm}
         \caption{Loan $S=4$}
         \label{fig:e2e_loan_4}
     \end{subfigure}
     \begin{subfigure}[b]{0.24\textwidth}
         \centering
         \includegraphics[width=\textwidth]{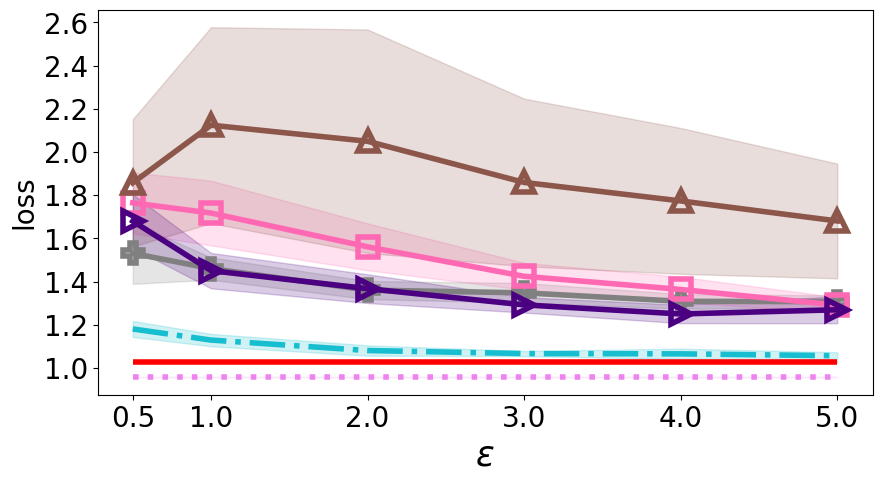}
         \vspace{-0.4cm}
         \caption{Letter $S=4$}
         \label{fig:e2e_letter_4}
     \end{subfigure}
    \vspace{-0.1cm}
    \caption{\kmeans loss with final $k$ centers. }
    \vspace{-0.1cm}
    \label{fig:e2e_loss}
\end{figure*}

\begin{figure*}
    \centering
    \begin{subfigure}[b]{0.3\textwidth}
        \centering
        \includegraphics[width=\textwidth]{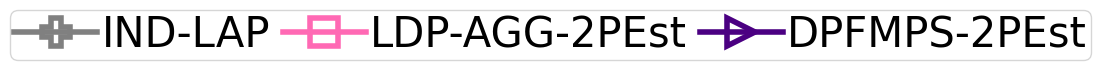}
    \end{subfigure}
    \\
     \begin{subfigure}[b]{0.24\textwidth}
         \centering
         \includegraphics[width=\textwidth]{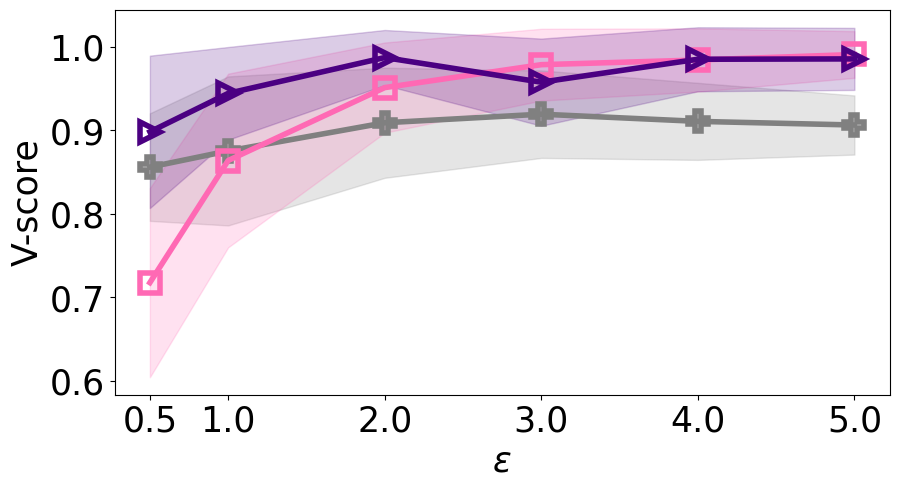}
         \vspace{-0.4cm}
         \caption{Mixed Guassian $S=2$}
         \label{fig:label_mg_2}
     \end{subfigure}
     \begin{subfigure}[b]{0.24\textwidth}
         \centering
         \includegraphics[width=\textwidth]{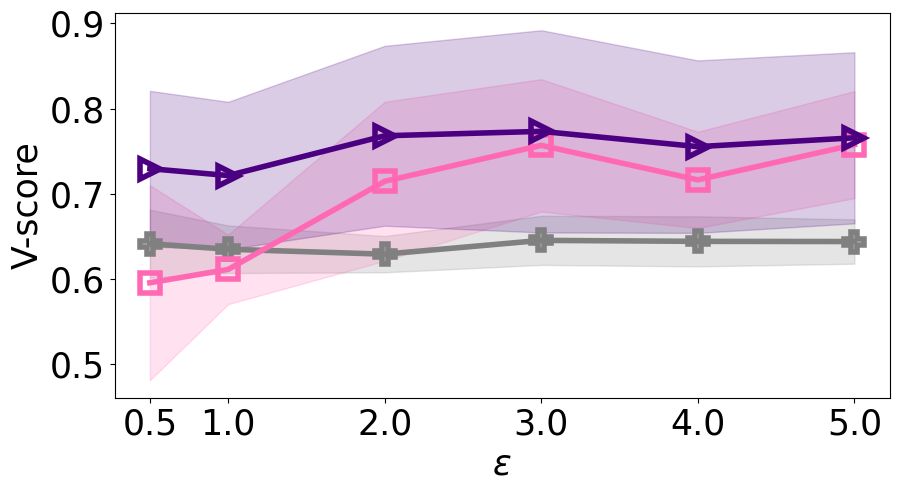}
         \vspace{-0.4cm}
         \caption{Taxi $S=2$}
         \label{fig:label_taxi_2}
     \end{subfigure}
     \begin{subfigure}[b]{0.24\textwidth}
         \centering
         \includegraphics[width=\textwidth]{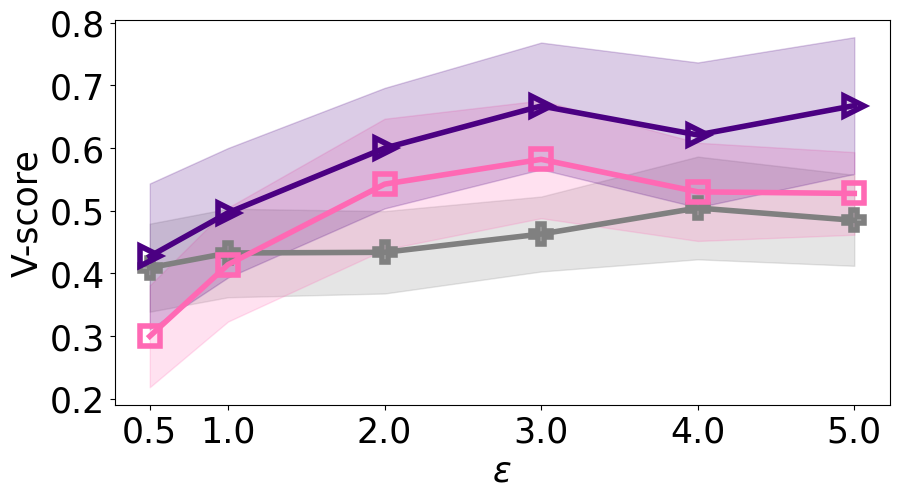}
         \vspace{-0.4cm}
         \caption{Loan $S=2$}
         \label{fig:label_loan_2}
     \end{subfigure}
     \begin{subfigure}[b]{0.24\textwidth}
         \centering
         \includegraphics[width=\textwidth]{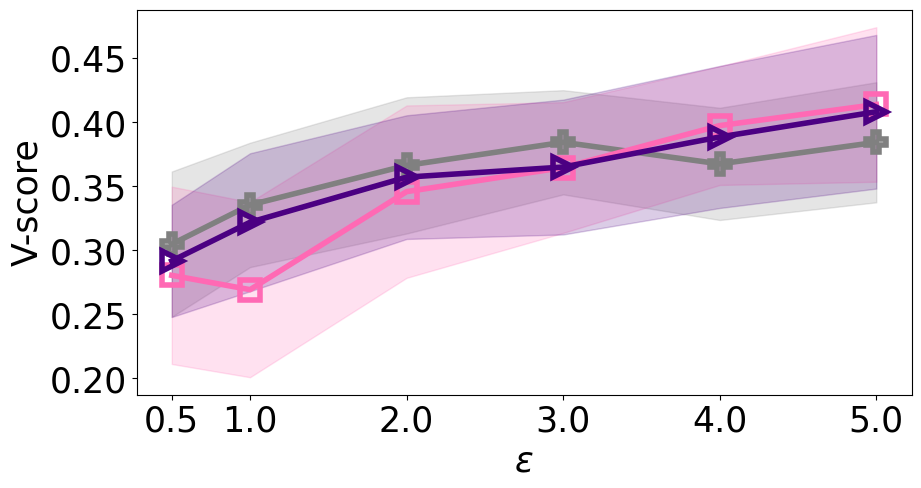}
         \vspace{-0.4cm}
         \caption{Letter $S=2$}
         \label{fig:label_letter_2}
     \end{subfigure}
     \\
     \begin{subfigure}[b]{0.24\textwidth}
         \centering
         \includegraphics[width=\textwidth]{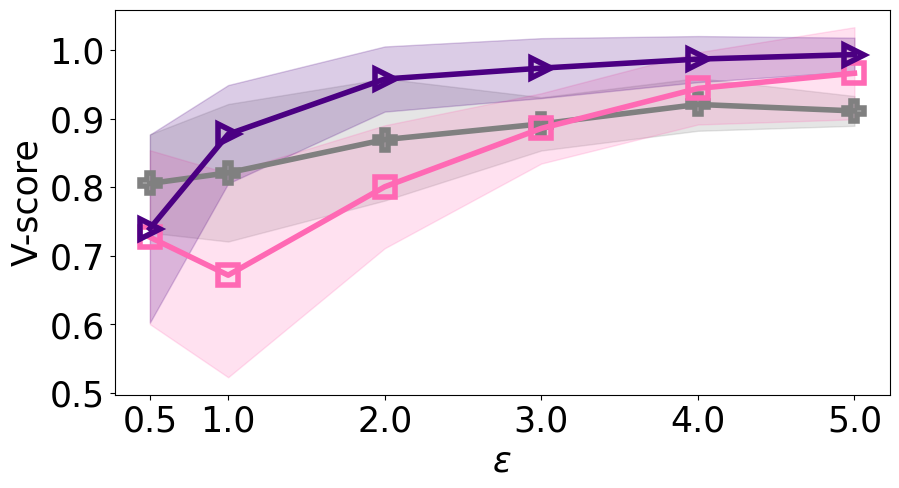}
         \vspace{-0.4cm}
         \caption{Mixed Guassian $S=4$}
         \label{fig:label_mg_4}
     \end{subfigure}
     \begin{subfigure}[b]{0.24\textwidth}
         \centering
         \includegraphics[width=\textwidth]{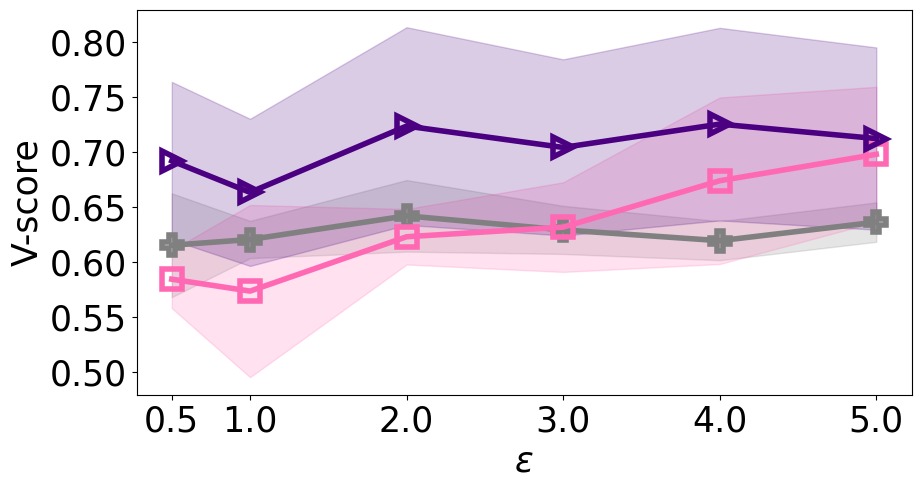}
         \vspace{-0.4cm}
         \caption{Taxi $S=4$}
         \label{fig:label_taxi_4}
     \end{subfigure}
     \begin{subfigure}[b]{0.24\textwidth}
         \centering
         \includegraphics[width=\textwidth]{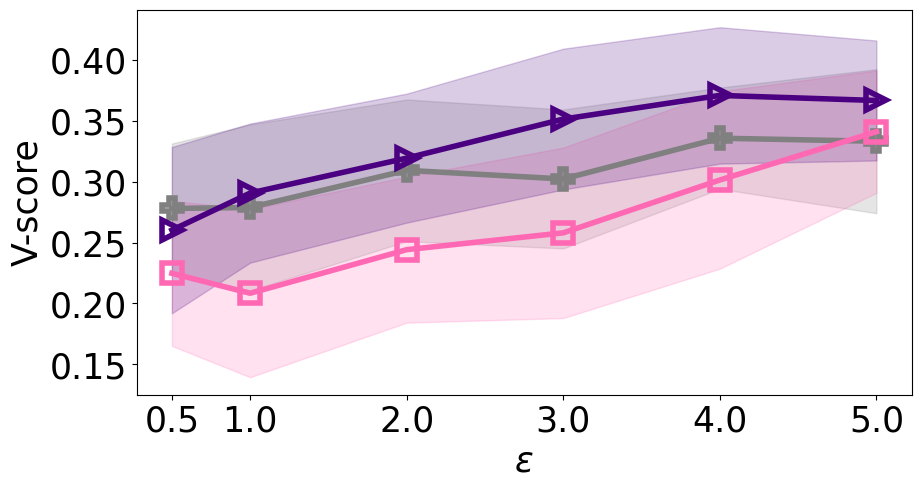}
         \vspace{-0.4cm}
         \caption{Loan $S=4$}
         \label{fig:label_loan_4}
     \end{subfigure}
     \begin{subfigure}[b]{0.24\textwidth}
         \centering
         \includegraphics[width=\textwidth]{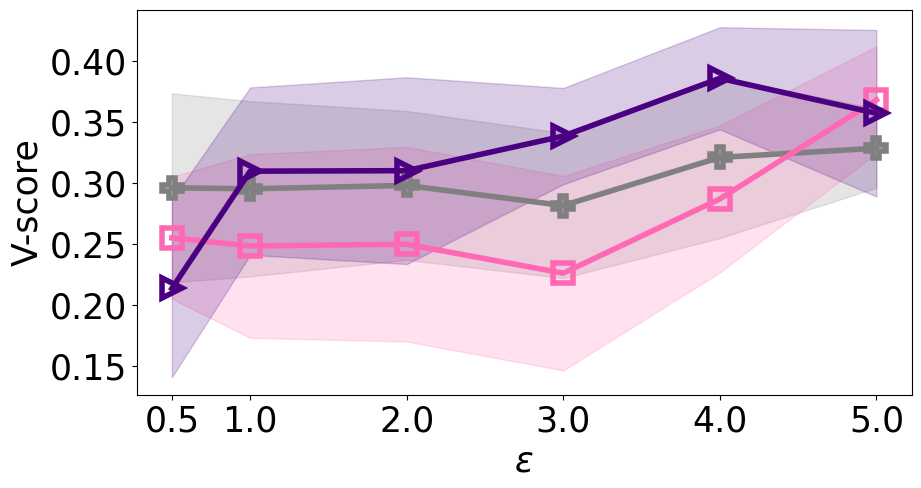}
         \vspace{-0.4cm}
         \caption{Letter $S=4$}
         \label{fig:label_letter_4}
     \end{subfigure}
    \vspace{-0.1cm}
    \caption{V-scores with final $k$ centers. }
    \vspace{-0.1cm}
    \label{fig:label_score}
\end{figure*}

\mypara{Datasets.}
We use four different datasets in our experiments, preprocess the data so that all attributes are normalize to $[-1, 1]$.

\mysubpara{Synthetic mixed Gaussian dataset with $k$ centers.}
To echo the implicit assumption of \kmeans problem, we first synthesize a mixed Gaussian dataset with $m=8$ for evaluation.
We generated this data in a similar way as \cite{chang2021locally} but enforced each dimension's range to be $[-1, 1]$ instead of in a $L_2$ ball.
We first randomly sample $k=5$ centers in the domain, then randomly sample $n=20,000$ data points from the Gaussian distribution with the means at those $k$ centers.

\mysubpara{New York Taxi dataset~\cite{taxi}.}
This taxi dataset contains 8 attributes and 100,000 records of taxi trips information, including pick-up/drop-off times and locations.
We preprocess the pick-up/drop-off times to make them a number indicating the time in a week, ranging from 0 to $60 \times 60 \times 24 \times 7$ before normalizing them. 

\mysubpara{Loan dataset~\cite{loan}.}
The original Loan dataset has 120 attributes. 
We extract the first 60k records and 16 numerical attributes of the applicant's credit information and apartment information.
We clip these attributes' values and make them upper-bounded by their original 95\% quantile to eliminate the outliers.

\mysubpara{Letter dataset~\cite{letter}.}
This dataset consists of information about black-and-white rectangular pixels displayed as the letters in the English alphabet.
There are 16 attributes and 20,000 records in the datasets.

We randomly partition the attributes of the mixed Gaussian dataset and the Loan datasets into $S$ parties; we split the attributes of the Taxi and Letter dataset with high correlation into different parties to simulate the worst case of information loss in the VFL.

\mypara{Parameter settings.}
There are different methods to decide the best $k$ value for the real-world datasets.
The Silhouette method~\cite{rousseeuw1987silhouettes} is one of the most commonly used.
Silhouette measures how closely data points are matched to data within its cluster and how loosely they are matched to data of the neighboring cluster.
The datasets we use in this paper all have relatively high Silhouette scores when $k$ is smaller.
Thus, we fix $k=5$ for experiments in this paper to eliminate the effect of $k$ unless we explicitly mention it.

According to the experimental results in \cite{smith2020fmsketch}, a larger number of repeated FM sketches (hyper-parameter $M$ in our paper) leads to a smaller relative error of cardinality estimation.  
Thus, we set $M = 4096$ and $\delta=1/n$ as default in different experiments, making the communication and privacy cost reasonable in the cross-silo FL setting and achieving stable accuracy.

\mypara{Evaluation Metric.}
We evaluate the quality of the final output with two metrics. 
The first is the normalized central \kmeans loss with all the data points, i.e., $\frac{1}{n}\sum_{x \in \X} (\min_{c\in \centroids}\norm{x - c}^2)$, which measures \emph{how representative} the final centers are.
The second is the V-score~\cite{rosenberg2007v} in the scikit-learn package~\cite{scikit-learn}, which measures \emph{how consistent} the clustering results of VFL algorithms are when compared to the ground truth classes (for the synthetic dataset) or central non-private \kmeans results (for the real word datasets).
The V-score is the harmonic mean of two conditional entropy scores measuring the homogeneity and completeness.
It is a score between 0 and 1, and the closer to 1 the better.
Because of the space limitation, we refer readers to~\cite{rosenberg2007v} for detailed formulas.

\mypara{Compared methods.}
The adapted DPLSF is used as the \localclustering in all private VFL experiments. 
So we name the end-to-end private VFL clustering method based on their \membershipencoding and \weightestimate instantiations.
Our experiments compare the following methods:

\mypara{$\bullet$} our proposed method \underline{\dpfmpsest} and \underline{\dpfmpsimprove}\footnote{\url{https://anonymous.4open.science/r/public_vflclustering-63CD}};

\mypara{$\bullet$} two baseline methods \underline{\indca} and \underline{\ldpimprove} (\ldpimprove is an improved version of \ldpca using the same technique in Section~\ref{subsec:baselines} because directly applying the \ldpca gives us a very large loss when $S>2$);

\mypara{$\bullet$} the central private \kmeans method \underline{\textsf{(central) DPLSF}} from \cite{googledpclustering};

\mypara{$\bullet$} the \underline{\textsf{central non-private}} \kmeans from the scikit-learn package~\cite{scikit-learn};

\mypara{$\bullet$} a \underline{\textsf{VFL non-private}} implementation following~\cite{ding2016k}.

\subsection{End-to-end Comparison}
\label{subsec:e2e}

We first demonstrate the end-to-end performance of the algorithm and show the advantages of our proposed algorithms.

\mypara{\kmeans loss results.}
Figure~\ref{fig:e2e_loss} shows the loss of the final centers produced by different algorithms.
As we can see, the losses of non-private vertical federated \kmeans are higher than the that of the central version in most cases.
Notice that we set $k'=k=5$ for the non-private VFL algorithm, so there are $25$ or $625$ grid nodes when $S=2$ or $S=4$ accordingly.
It means some information is lost when we building the grid, and a more fine-grained grid can improve the utility in the non-private setting if we compare the figures in the first rows with the ones in the second row.
In some cases, the central DPLSF outperforms the non-private VFL baseline because the noise with a large privacy budget has a smaller impact on the final results than the information loss of representing the local data with local centers and memberships.

Comparing our proposed method, \dpfmpsimprove, with the two baseline methods and \dpfmpsest, we can see that our proposed method can produce final centers with consistently lower losses for most settings.
When there are only two parties ($S=2$), \dpfmpsest has the same performance as \dpfmpsimprove because they are the same. 
But \dpfmpsimprove largely improves over \dpfmpsest when $S=4$.
The exception is the outcomes of the \indca on the Letter dataset (Figure~\ref{fig:e2e_letter_2} and \ref{fig:e2e_letter_4}), where \indca approach has slightly better performance when $\epsilon$ is small. 
That is because the attributes of the Letter dataset are relatively independent and do not follow the underlying assumption of k-means algorithm. 
As for the LDP baseline, \ldpimprove,  we can see that its performance is close to our proposed method only with large privacy budget but is inferior when $\epsilon$ is small, which echoes with the theoretical results in Section~\ref{sec:membership}.
When the privacy budget is large enough, the LDP protocol has little randomness to perturb the local membership of a user. 
In contrast, the FM sketch has its inherent randomness, even with large privacy budgets.
However, one may also notice that the empirical losses are closer to the non-private ones than the theorem indicates.
It is because the accuracy of FM sketches is usually better than its theoretical guarantee, especially after normalization (Line~\ref{algstep:consistent} of Algorithm~\ref{algo:dpfms-est}) \cite{smith2020fmsketch, lang2017back}.

Moreover, comparing the private central with the private VFL algorithms, we can see that the private VFL almost always has a higher cost than the central one.
The gap exists because of two different reasons.
1) Only the local centers and the sketches are shared with the central server in VFL, so some information is lost.
2) When we use sketches to encode the cardinalities of the intersections, privacy budgets are split to both differential private local clustering and generating FM sketches.

\mypara{V-scores.}
When generating the synthetic mixed Gaussian dataset, we assign the same label for the data points drawn from the same center.
However, the real-world datasets have no label, or the number of labels is different from our experiment setting, so we apply the central non-private \kmeans algorithm to generate the labels for the dataset.
Thus, the closer the score is to 1, the more similar the clustering result is to the real labels (for Mixed Gaussian) or the central non-private \kmeans clustering (Taxi, Loan and Letter).

Figure~\ref{fig:label_score} presents the result.
The results align with the \kmeans loss results. In most settings, we can observe that \dpfmpsimprove outperforms the other two VFL private baseline methods in Figure~\ref{fig:label_score}.
For the Letter dataset, we can see that all three methods have lower scores than the ones of other datasets when $S=2$.
It is because the data in Letter have very independent attributes and bring advantages to the \indca.
However, \dpfmpsimprove still outperforms other methods on other datasets with different privacy budgets, and it also has performance close to best when $S=2$.

\begin{figure}
    \centering
    \begin{subfigure}[b]{0.25\textwidth}
        \centering
        \includegraphics[width=\textwidth]{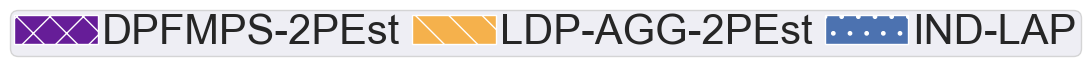}
    \end{subfigure}
    \\
     \begin{subfigure}[b]{0.23\textwidth}
         \centering
         \includegraphics[width=\textwidth]{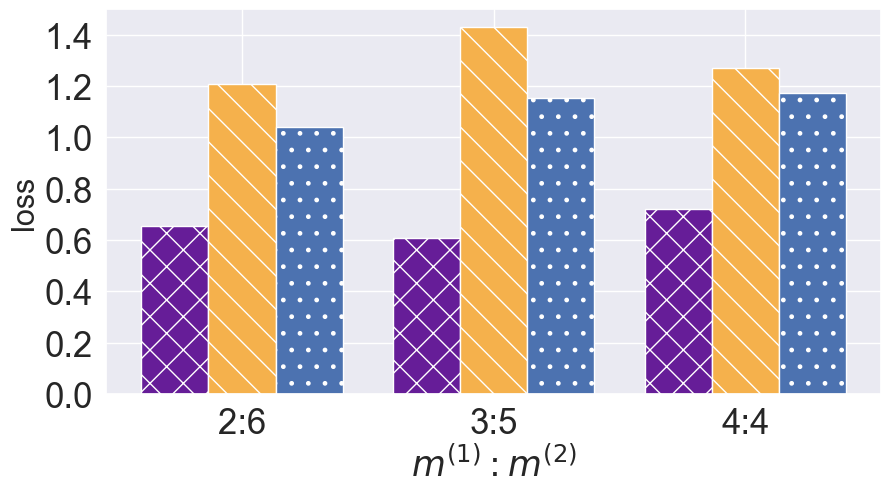}
         \vspace{-0.5cm}
         \caption{Mixed Guassian $\epsilon=1$}
         \label{fig:uneven_mg}
     \end{subfigure}
     \begin{subfigure}[b]{0.23\textwidth}
         \centering
         \includegraphics[width=\textwidth]{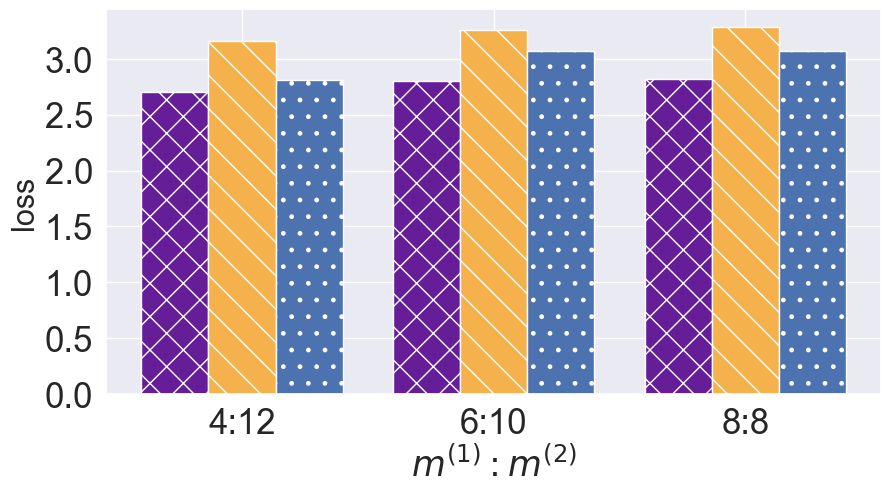}
         \vspace{-0.5cm}
         \caption{Loan $\epsilon=1$}
         \label{fig:uneven_loan}
     \end{subfigure}
     \vspace{-0.2cm}
    \caption{Effect of unevenly split dataset.}
    \vspace{-0.2cm}
    \label{fig:uneven}
\end{figure}

\begin{figure}
    \centering
    \begin{subfigure}[b]{0.25\textwidth}
        \centering
        \includegraphics[width=\textwidth]{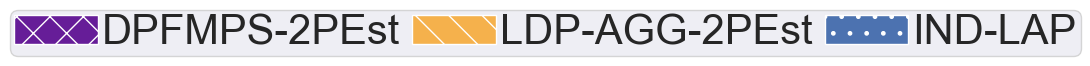}
    \end{subfigure}
    \\
     \begin{subfigure}[b]{0.23\textwidth}
         \centering
         \includegraphics[width=\textwidth]{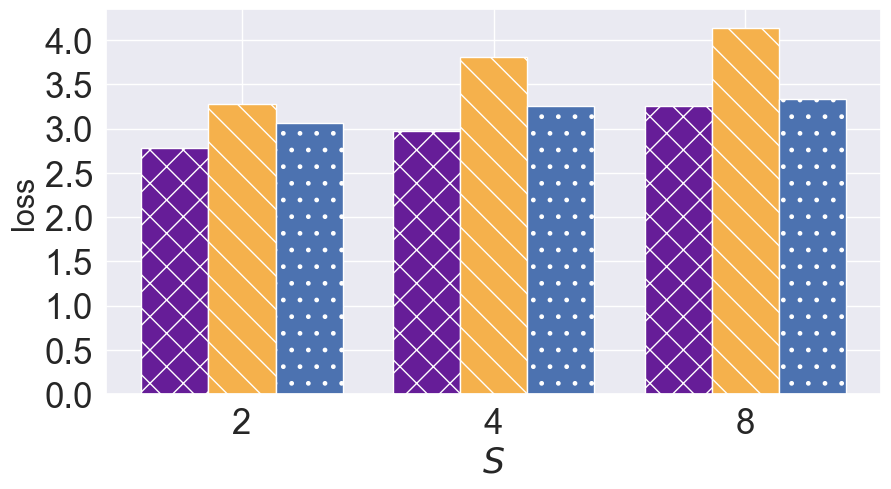}
         \vspace{-0.5cm}
         \caption{Loan $\epsilon=1$}
         \label{fig:s_eps_1}
     \end{subfigure}
     \begin{subfigure}[b]{0.23\textwidth}
         \centering
         \includegraphics[width=\textwidth]{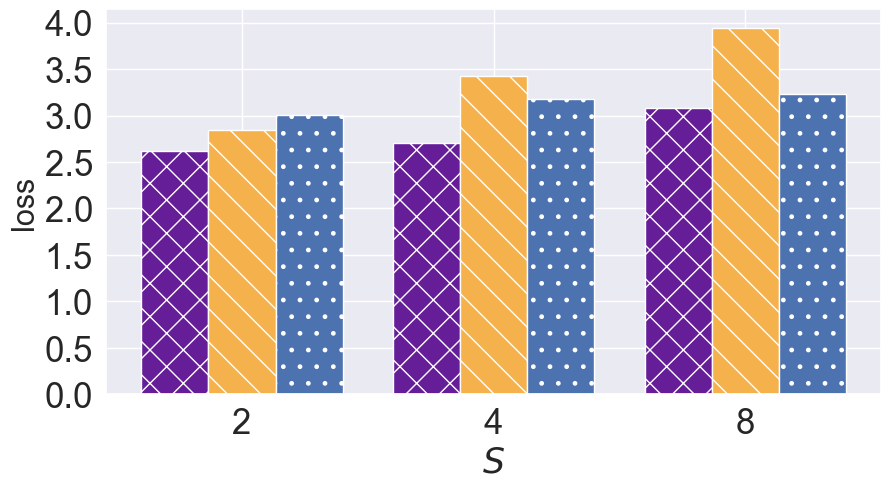}
         \vspace{-0.5cm}
         \caption{Loan $\epsilon=2$}
         \label{fig:s_eps_2}
     \end{subfigure}
     \vspace{-0.2cm}
    \caption{Effect of different $S$ on Loan dataset ($m=16$).}
    \vspace{-0.2cm}
    \label{fig:compare_s}
\end{figure}

\mypara{Effect of uneven number of attributes.}
The previous results show how \dpfmpsimprove outperforms other baselines method when the attributes are evenly split.
We also explore how the methods perform when the each party has different number of attributes.
Denote $m\rel{}{i}$ as the number of attributes in the $i$-th party's dataset. 
We split the mixed Gaussian dataset so that $m\rel{}{1} : m\rel{}{2}$, is $3:5$ or $2:6$, and split the Loan dataset to $6: 10$ and $4: 12$.
The results shown in Figure~\ref{fig:uneven} indicate that our \dpfmpsimprove work consistently well in the sense that the losses of different split ratios do not change much.
Besides, \dpfmpsimprove losses are smaller than the other two baselines in the figure.

\mypara{Effect of number of parties $S$.}
We also compare how the number of parties $S$ can affects the final clustering results in Figure~\ref{fig:compare_s}. The loss generally increases as $S$ increases.
This is mainly because the \weightestimate component introduces a larger error because of the random noise for privacy.
However, \dpfmpsimprove always provides the lowest loss among the three.

\subsection{Ablation Study of Components}
We perform the following ablation studies to demonstrate the impact of enforcing privacy on different components.

\mypara{Intersection cardinality accuracy comparison.}
To break-down the error, the first interesting metric is the relative accuracy of the intersection cardinality estimation.
The relative error is defined as 
$\frac{1}{n}\sum_{(a_1, \ldots, a_S)} \norm{w(\grid_{(a_1, \ldots, a_S)}) - w^*(\grid_{(a_1, \ldots, a_S)})}_1$.

We fix $k=k'=5$ for fair comparison of all methods and the results are shown as Figure~\ref{fig:ca-estimate-error}.
Because the evaluation is based on the intermediate results of the end-to-end private algorithm, only half of the privacy budgets are spend on the intersection estimation.
We can see that our proposed method \dpfmpsimprove can outperform other methods in most experiment settings.
The \dpfmpsest performs similar as \dpfmpsimprove in the $S=2$ setting as expected.
\dpfmpsimprove can significantly improve the accuracy when there are more parties ($S=4$).
The LDP baseline \ldpimprove still has relative error larger than the \dpfmpsimprove, even with the 2-way iterative updating.
Moreover, as we can see from the figure, the error of the 1-way approach \indca is dominated by loss of dependency information between the attributes and barely going down as we increase the privacy budget.

\begin{figure}
     \centering
     \begin{subfigure}[b]{0.4\textwidth}
        \centering
        \includegraphics[width=\textwidth]{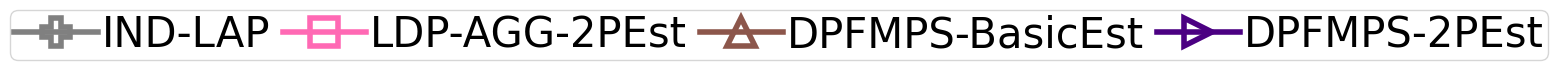}
    \end{subfigure}
    \\
     \begin{subfigure}[b]{0.23\textwidth}
        \centering
        \includegraphics[width=\textwidth]{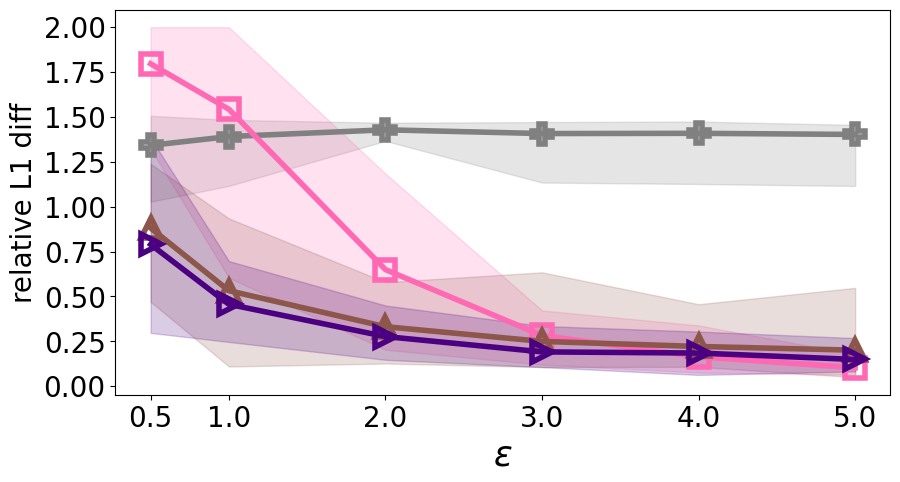}
        \vspace{-0.4cm}
        \caption{Mixed Gaussian $S=2$}
        \label{fig:mg_ca_2way_S=2}
    \end{subfigure}
    \begin{subfigure}[b]{0.23\textwidth}
        \centering
        \includegraphics[width=\textwidth]{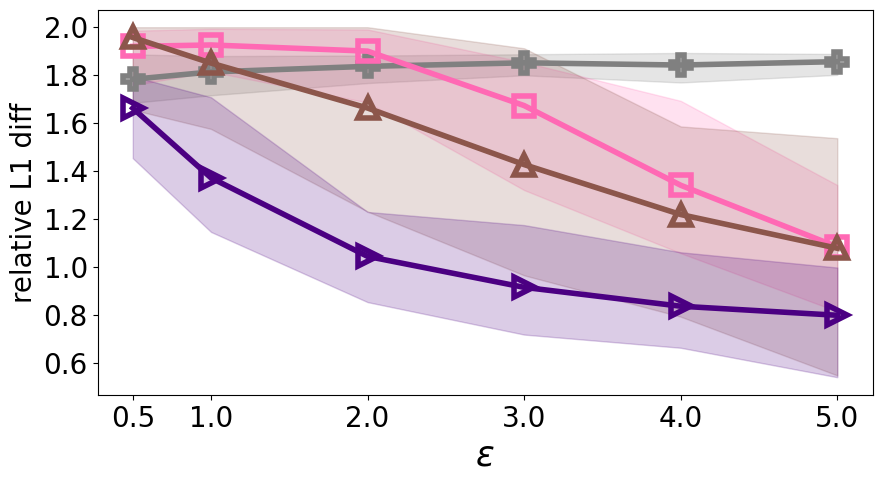}
        \vspace{-0.4cm}
        \caption{Mixed Gaussian $S=4$}
        \label{fig:mg_ca_2way_S=4}
     \end{subfigure}
    \vspace{-0.2cm}
     \caption{Errors of estimate the intersection cardinalities.}
     \vspace{-0.2cm}
     \label{fig:ca-estimate-error}
\end{figure}

\begin{figure}
     \centering
     \begin{subfigure}[b]{0.48\textwidth}
        \centering
        \includegraphics[width=\textwidth]{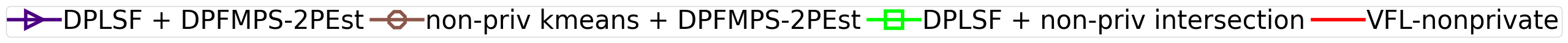}
    \end{subfigure}
    \\
     \begin{subfigure}[b]{0.23\textwidth}
        \centering
        \includegraphics[width=\textwidth]{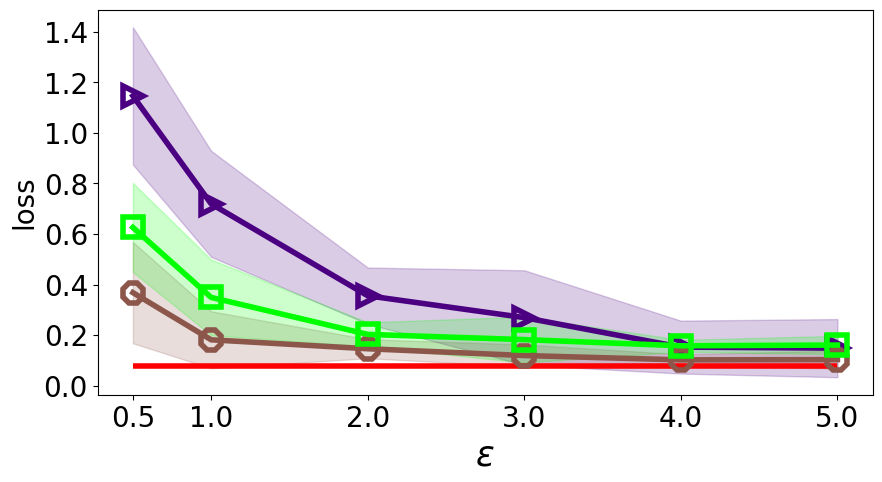}
        \vspace{-0.4cm}
        \caption{Mixed Gaussian $S=2$}
        \label{fig:mg_break_S=2}
    \end{subfigure}
    \begin{subfigure}[b]{0.23\textwidth}
        \centering
        \includegraphics[width=\textwidth]{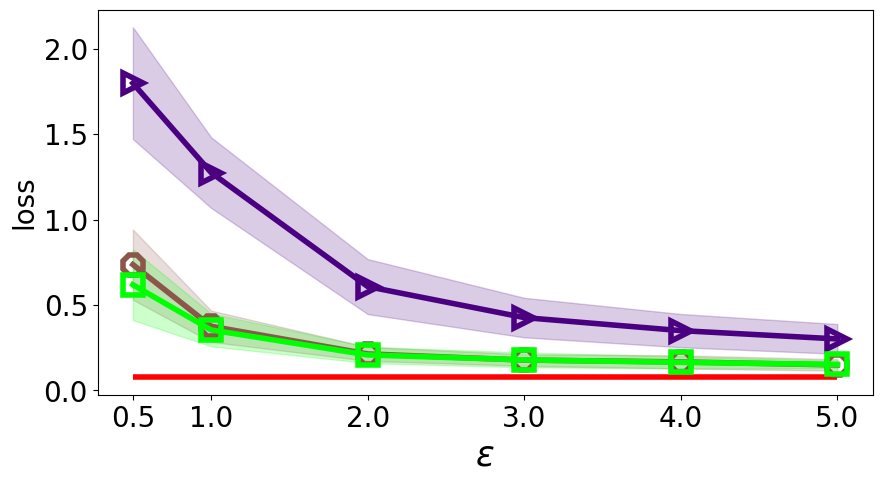}
        \vspace{-0.4cm}
        \caption{Mixed Gaussian $S=4$}
        \label{fig:ca_2way_S=4}
    \end{subfigure}
    \vspace{-0.2cm}
     \caption{Impact of enforcing privacy on components.}
     \label{fig:break-error}
     \vspace{-0.2cm}
\end{figure}

\mypara{Distinguishing the impact of private \localclustering from private intersection cardinality estimation.}
In Figure~\ref{fig:break-error}, we compare the impact of enforcing privacy on either the local clustering component or the intersection estimation component.
For ``\texttt{non-priv kmeans + \dpfmpsimprove}'' experiments, each data party uses non-private central \kmeans to generate local centers and spends $\epsilon_2$ on the intersection estimation algorithm.
For ``\texttt{DPLSF + non-priv intersection}'' experiments, each data party spends $\epsilon_1$ privacy budget on DPLSF to generate local centers, and the intersection cardinalities are computed exactly without enforcing any privacy.

As shown Figure~\ref{fig:break-error}, enforcing the \localclustering with DP but using the non-private intersection cardinality estimation gives the cost closer to the end-to-end private ones when $S=2$.
However, when $S=4$, making either component non-private while enforcing the other private has a similar effect on the final loss.
These results show that when $S$ is small, the private \localclustering (related to the first term of additive error $\lambda$ of Theorem~\ref{thm:final-utility}) is the component that introduces the majority error.
However, they also support our analysis in Section~\ref{subsec:utility} that the intersection component will introduce a larger error when the number of parties $S$ increases.

\begin{figure}
    \centering
    \begin{subfigure}[b]{0.3\textwidth}
        \centering
        \includegraphics[width=\textwidth]{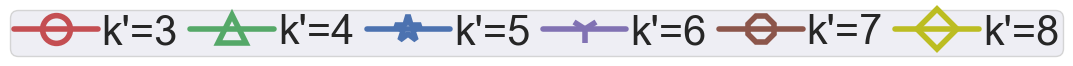}
    \end{subfigure}
    \\
    \begin{subfigure}[b]{0.23\textwidth}
        \centering
        \includegraphics[width=\textwidth]{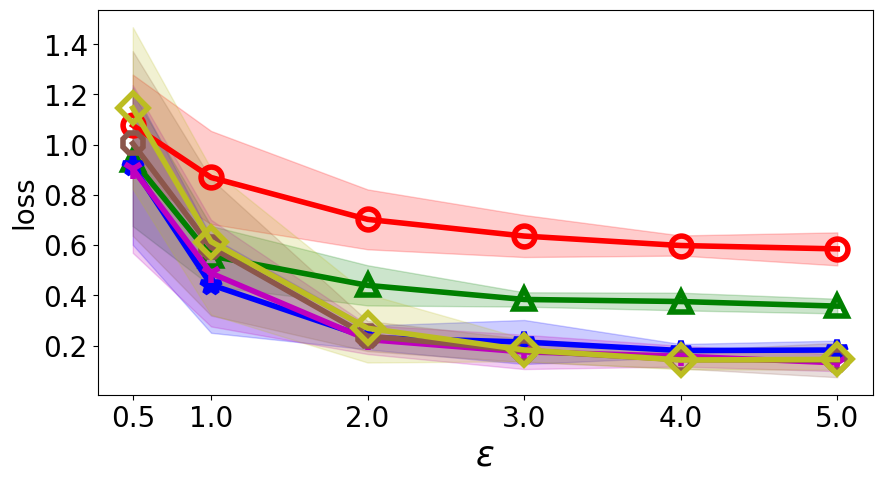}
        \vspace{-0.2cm}
        \caption{Mixed Gaussian, $S=2$}
        \label{fig:mg_localk2}
    \end{subfigure}
    \begin{subfigure}[b]{0.23\textwidth}
        \centering
        \includegraphics[width=\textwidth]{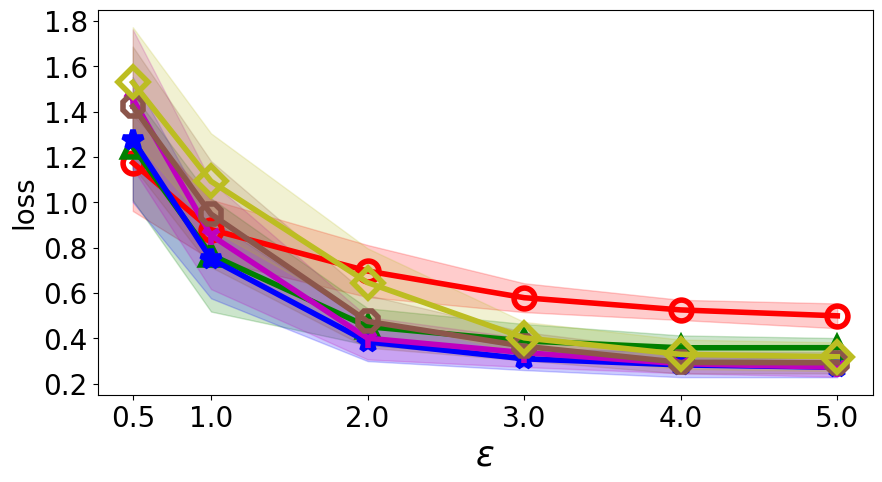}
        \vspace{-0.2cm}
        \caption{Mixed Gaussian, $S=4$}
        \label{fig:mg_localk4}
    \end{subfigure}
    \vspace{-0.2cm}
    \caption{Different local $k'$ with different privacy budget.}
    \label{fig:localk}
    \vspace{-0.2cm}
\end{figure}

\mypara{Impact of different local $k'$.}
\label{subsec:exp_vary_local_k}
As we discussed in Section~\ref{subsec:vary-local-k}, in order to trade-off between the error introduced in the membership intersection cardinality estimation with the information loss of local clustering, we can adjust the number of local clusters, namely $k'$.
In Figure~\ref{fig:localk}, we compare the impact of different numbers of local clusters on the final cost.
The empirical optimal $k'$ for different privacy budgets are close to our heuristic choice.
If $\epsilon$ is large, larger $k'$ can benefit the final result by keeping more local information; if $S$ is large, then it would be better to choose a smaller $k'$ to limit the exponentially grown number of grid nodes.

\mypara{Communication/computation cost.}
The communication cost of each data party in our method is determined by two parameters, the number of local clusters $k'$ and the repetition of FM sketches $M$.
We record the communication cost of some combinations in Table~\ref{table:commcost} with the mixed Gaussian dataset.
As we can see, \indca has the smallest communication cost, at the expense of ignoring all inter-party correlation. 
Our method can have a smaller communication cost than the \ldpimprove and the non-priv if $n > Mk'$.
With larger $k'$ and $S$ (e.g., $k'=8$ and $S=4$), the \dpfmpsimprove and \ldpimprove require more iterations to achieve the convergence of the estimates.
This iterative update dominates the computation time, but the computation can still be done efficiently.

\begin{table}
\caption{Communication and computation cost comparison.}
\resizebox{0.90\columnwidth}{!}{%
\begin{tabular}{c|c|c|c|c}
\toprule
\multirow{2}{*}{$k'$}  &\multirow{2}{*}{method} & \multirow{2}{*}{\begin{tabular}{c} Comm cost\\(per party) \end{tabular}} & \multicolumn{2}{c}{Grid compute time} \\
& & &$S=2$ & $S=4$\\
\midrule
\multirow{5}{*}{$5$}  
& \dpfmpsimprove($M=2048$) & $\approx 82$ kB & 4.75 s & 35.40 s  \\
& \dpfmpsimprove($M=4096$) & $\approx 164$ kB & 7.47 s & 39.08 s  \\
& \indca & $\approx 0.19$ kB & 0.61 s & 1.32 s\\
& \ldpimprove & $\approx 160$ kB & 5.46 s & 38.48 s\\
& non-priv & $\approx 160$ kB & 0.15 s & 1.12 s\\
\midrule
\multirow{5}{*}{$8$}  
& \dpfmpsimprove($M=2048$) & $\approx 131$ kB & 11.68 s & 98.91 s  \\
& \dpfmpsimprove($M=4096$) & $\approx 260$ kB & 11.95 s & 101.14 s \\
& \indca & $\approx 0.28$ kB & 0.64 s & 3.83 s\\
& \ldpimprove & $\approx 160$ kB & 12.13 s & 100.63 s \\
& non-priv & $\approx 160$ kB & 0.19 s & 3.26 s\\
\bottomrule
\end{tabular}
}
\vspace{-0.1cm}
\label{table:commcost}
\end{table}

\section{Related work}
\label{sec:related}
We summarize the related work from the following three axes.

\mypara{Vertical federated learning.} 
To our knowledge, this paper is the first work that targeting at the clustering problem under vertical federated learning with DP guarantee.
VFL has been studied in the recent years, sometimes under the name of vertical distributed learning.
The most relevant paper is \cite{ding2016k}, which proposes a solution for the clustering problem but does not consider the privacy leakage issue.
Existing work about other problems in the VFL setting includes learning tree models with secure multiparty computation techniques~\cite{wu13pivot, liu2020federated-forest} and training composed models~\cite{hu2019fdml}.
Some recent papers \cite{hu2019admm-vertical, xie2022improving} apply the ADMM framework in the VFL setting.
Another paper~\cite{chen2020vafl} discussing asynchronous supervised learning with VFL and DPSGD~\cite{abadi2016deep} assumes the labels are public accessible.

\mypara{Data sketches with DP and private set intersection cardinality estimation.}
It is originally claimed in \cite{desfontaines2019cardinality} that cardinality estimators, including FM sketch as its variants, leak the membership of a user in a set.
However, their result assumes that the adversary knows the hash key, which is not a common DP security setting.
Because DP security is based on the assumption that the adversary does not know the random seed; otherwise the adversary can re-generate the randomness (i.e., Laplace noise) by itself and recover the true value. 
Recently, several papers \cite{smith2020fmsketch, hu2021ca, dickens2022allsketch, choi2020dpsketch} reveal that hash-based, order-invariant sketches satisfy differential privacy as long as the cardinality set is large enough.
An earlier paper \cite{pagh2021linear} uses linear sketch and perturbs the sketch with random response, but it has larger relative error~\cite{smith2020fmsketch} and assumes no duplicate.

While the private set intersection cardinality (PSI-CA) problem is a traditional cryptographic problem and has been studied in series of literature~\cite{cristofaro2012fastPSICA, freedman2004efficientPSI, kissner2005privacy, hazay2008efficient, jarecki2009efficient, hazay2010efficient}, there are few papers about how to solve it under DP.
To solve the PSI-CA problem with DP guarantee, the existing solutions also rely on some kinds of data sketch.
For example, \cite{stanojevic2017distributed} proposes a solution to support set union or intersection  with an untrusted third party based on Bloom filters.
However, their solution introduces large randomness and is unable to scale to the setting with more than 2 parties.
Other similar work includes using cryptographic tools to encrypt the sketches as \cite{kreuter2020privacy}.
Some other cryptography oriented research~\cite{groce2019cheaper, kacsmar2020dppsi} use DP definition as a relaxation of the traditional security definition, and develop secure PSI protocol resisting malicious adversary.

\mypara{Differential private \kmeans.}
In the central DP setting, some earlier papers, including~\cite{stemmer2018differentially, nissim2018clustering, huang2018optimal, blum2005practical, nissim2007smooth, feldman2009private, wang2015differentially, nissim2016locating}, contribute to the theoretical bound for the \kmeans cost.
A practical implementation \cite{balcan2017kmeans} with cost bounded is based on privately selecting a candidate center set and gradually swapping in better centers into the final $k$ centers.
Another open-source implementation of the DP \kmeans is \cite{googledpclustering}, which is based on locality sensitive hashing.
Its adapted version is used in this paper as a building block.

\section{Conclusion}
\label{sec:conclustion}
In this paper, we propose novel differentially private solutions for the vertical federated clustering problem.
We demonstrate that our solution can outperform other baselines while providing desired privacy protection on the local data.
Some future directions include extending and customizing the DP intersection cardinality estimation sketches to other VFL problems, and providing solutions that can support more data parties and larger $k$ at once.

\begin{acks}
This work is supported in part by the United States NSF under Grant No. 2220433, No. 2213700, No. 2217071.
\end{acks}

{
\bibliographystyle{abbrv}
\bibliography{reference}
}
\appendix
\clearpage
\section{Additional Proofs}
\label{app:proof}

\subsection{Privacy proofs}
\begin{proof}[Proof of Lemma~\ref{lemma:single}]
To prove the privacy guarantee of Algorithm~\ref{algo:dpfmps}, we denote a set of identities as $\members=\left\{\id_1, \ldots, \id_n \right\}$, where $\id$ is not necessarily a integer, but any input that is hash-able.
For clustering, $\members$ is divided into $\{\members\rel{1}{\ell}, \ldots, \members\rel{k'}{\ell}\}$ according to the differentially private $\centroids\rel{}{\ell}$ from the previous phase:
For any $\id$, the partition index is $\argmin_{j\in [k']}\norm{x\rel{\id}{\ell} - c\rel{j}{\ell}}^2_2$.
Now given a pair of neighboring datasets $\members$ and $\members'$, where there is an $\id$ such that $\id' \not \in \members$ but $\id_{i'} \in \members'$, their clustering memberships, $\left\{\members_{1}, \ldots, \members_{k'} \right\}$ and $\left\{\members'_{1}, \ldots, \members'_{k'} \right\}$, differing at exactly one partition $ j \in [k']$, such that $\id \in \members'_{j}, \id \not \in \members_{j}$; 
for all other subsets $\forall j'\neq j, \members_{j'} = \members'_{j'}$.
Thus, the following proof can be considered as following the parallel composition property.
\end{proof}

\subsection{Utility proofs}
To simplify the description, we denote $\phi_{\centroids\rel{}{\ell}}(\cdot)$
as a partition function mapping a user record in the local dataset of party $\ell$ to the closest center's index in $\centroids\rel{}{\ell}=\{c\rel{1}{\ell}, \ldots, c\rel{k'}{\ell}\}$ generated in the previous phase, i.e., $\phi_{\centroids\rel{}{\ell}}\left(x\rel{\id}{\ell}\right) = \argmin_{a\in [k']}\norm{x\rel{\id}{\ell} - c\rel{a}{\ell}}^2_2$.
With the partition function, each data party partitions the users into disjoint sets  $\{\members\rel{1}{\ell}, \ldots, \members\rel{k'}{\ell}\}$ such that user $\id \in \members\rel{a}{\ell}$ if and only if $a = \phi_{\centroids\rel{}{\ell}}\left(x\rel{\id}{\ell}\right)$.

\begin{proof}[Proof of Lemma~\ref{lemma:expectation+variance}]
Consider the output of the geometric-valued hash function $H$, $Z = \max\{Y_1, \ldots, Y_{\tilde{n}}\}$, where $Y_i\sim \geometric(\frac{\gamma}{1+\gamma})$.
The cumulative mass function of $Z$ is
\begin{align*}
    \Pr{Z \leq z} = \left(1 - \left(1 - \frac{\gamma}{1 + \gamma} \right)^z \right)^{\tilde{n}} = \left(1 - \frac{1}{(1 + \gamma)^z} \right)^{\tilde{n}} .
\end{align*}
The probability $\Pr{Z\leq \alpha_{\min}} = \frac{1}{e^{\epsilon'\tilde{n}}}$.
Based on this result, the ratio of expectation can be written as 
\begin{align*}
    \frac{\EV{\alpha}}{\EV{\hat{\alpha}}} =& 1 - \frac{\sum_{z\leq \alpha_{\min}}\Pr{Z = z} (\alpha_{\min} - z)}{\EV{\hat{\alpha}}} \\
    \geq& 1 - \frac{\sum_{z\leq \alpha_{\min}}\Pr{Z = z} \alpha_{\min} }{\EV{\hat{\alpha}}} \geq 1 - \frac{\log_{1+\gamma} (1 / (1 - e^{-\epsilon'}))}{e^{\epsilon'\tilde{n}}\EV{\hat{\alpha}}}
\end{align*}
Let $\tilde{n} = |\members| + n_p$.  When $|\members|$ is large enough, the denominator dominants and the value will go to 0 very quickly.
Similar analysis can also be applied to the variance.
\end{proof}

\begin{proof}[Proof of Theorem~\ref{thm:final-utility}]
For simplicity, we denote the virtual global dataset as $\X = [\X\rel{}{1}|\ldots|\X\rel{}{S}]$, the optimal global centers as $\centroids^*=\{c^*_1, \ldots, c^*_k\}$, and the membership of grid $g$ as $\members_g$
We also denote $\phi^*(\cdot)$ as the optimal partition function mapping a data point to a closest center index.
We first can bound the $\cost{\X}$ as the following:
    \begin{align*}
        \cost{\X}(\centroids) &= \sum_{i\in [n]}\norm{x_i - c_{\phi*(x_i)}}^2_2 \\
        &\leq \sum_{g\in \grid}\sum_{i \in \members_{g}} \norm{x_i - g + g - c_{\phi*(g)}}^2_2 \\
        & \leq 2\sum_{g\in \grid}\sum_{i \in \members_{g}} \norm{x_i - g}_2^2 + \norm{g - c_{\phi*(g)}}_2^2
    \end{align*}
    We first analysis the first term. 
    Notice that every local data party runs a $(\beta_{priv}, \lambda_{priv})$-approximate algorithm, and $g$ consists of $S$ local centers $[c\rel{a_1}{1}\mid \ldots \mid c\rel{a_S}{S}]$
    \begin{align*}
        \sum_{g\in \grid}\sum_{i \in \members_{g}} \norm{x_i - g}_2^2 &= \sum_{\ell\in [S]} \sum_{i\in[n]} (x\rel{i}{\ell} - c\rel{a_\ell}{\ell})^2  \\
        &\leq \sum_{\ell\in [S]} (\beta_{priv} \opt\rel{k, \X\rel{}{\ell}}{\ell} + \lambda_{priv}) \\
        &= \beta_{priv}\sum_{\ell\in [S]}\opt\rel{k, \X\rel{}{\ell}}{\ell} + 
        S \lambda_{priv} \\
        &\leq \beta_{priv} \opt_{k, \X} + 
        S \lambda_{priv}
    \end{align*}
    
    Notice that $w(g)$ is only an estimation of $|\members_g|$ in our algorithm.
    So we denote $\kappa = \sum_{g \in \grid} \left|w(g)- |\members_g| \right|$.
    The second term can be relaxed as the following:
    \begin{align*}
        \sum_{g\in \grid}\sum_{i \in \members_{g}} \norm{g - c_{\phi*(g)}}_2^2
        &= \sum_{g \in \grid} |\members_g| \norm{g - c_{\phi*(g)}}_2^2 \\
        &\leq \sum_{g \in \grid} |w(g)| \norm{g - c_{\phi*(g)}}_2^2 + 4m^2 \kappa
    \end{align*}
    
    As the central server finally runs an $(\beta_0, 0)$-approximate \kmeans algorithm, we also denote $\hat{\centroids}=\{\hat{c}_1, \ldots, \hat{c}_k\}$ as the set of optimal $k$ centers for \kmeans given $(\grid, w(\grid)$ as the dataset.
    Thus, it follows that
    \begin{align*}
        &\quad \sum_{g \in \grid} |w(g)| \norm{g - c_{\phi*(g)}}_2^2  \\
        &\leq \beta_0 \sum_{g \in \grid} |w(g)| \norm{g - \hat{c}_{\phi*(g)}}_2^2\\
        &\leq \beta_0 \sum_{g \in \grid} |\members_g| \norm{g - \hat{c}_{\phi*(g)}}_2^2 + 4\beta_0 m^2 \kappa \\
        &\leq \beta_0 \sum_{g \in \grid} \sum_{i \in \members_{g}} \norm{g - x_i + x_i - c^*_{\phi*(x_i)}}_2^2 + 4\beta_0 m^2 \kappa \\
        & \leq 2\beta_0 \sum_{g \in \grid} \sum_{i \in \members_{g}} \norm{g - x_i}_2^2 + 2\beta_0 \sum_{g \in \grid} \sum_{i \in \members_{g}} \norm{ x_i - c^*_{\phi*(x_i)}}_2^2 + 4\beta_0 m^2 \kappa \\
    \end{align*}
    We have bounded the first term above and can reuse the result here.
    The second term is the exactly $2\beta_0 \opt_{k, \X}$ by definition.
    So when we put every thing together, we have
    \begin{align*}
        \cost{\X}(\centroids) &\leq  \left( 2\beta_{priv} + 4\beta_0 + 4 \beta_0 \beta_{priv} \right)\opt_{k, \X} \\
        &+ 2(\beta_0 + 1)S\lambda_{priv} \\
        &+ 8(\beta_0 + 1)m^2 \kappa
    \end{align*}
    
To bound the $t$, we can use Chebyshev's inequality with the union bound of all grid nodes and the results of Lemma~\ref{thm:sketch-std}, the standard  $\sigma = O\left(\frac{n}{\sqrt{M}} +  \frac{S(k-1)\sqrt{\log (1/\delta)}}{\epsilon_2}  \right)$.
Because there are totally $k^S$ intersection cardinality estimates in the grid, the probability of the odd scenario for each intersection cardinality estimate we want to bound is $\frac{\omega}{k^S}$ because of the union bound. 
According to the Chebyshev’s inequality, for one intersection estimation, $\Pr{|w(g) - |\mathcal{M}_g| | \geq \sigma \sqrt{\frac{k^S}{\omega}}} \leq \frac{\omega}{k^S}$.  Plug in the standard deviation result and consider the there are totally $k^S$ intersection, we can have the result shown in our theorem with $k^{1.5S}$.
Thus, with probability at least $1 - \omega$,
\begin{align*}
    \kappa &\leq \frac{k^{1.5S}}{\sqrt{\omega}} \sigma = O\left(\frac{n k^{1.5S}}{\sqrt{\omega M}} +  \frac{S(k-1)k^{1.5S}\sqrt{\log (1/\delta)}}{\epsilon_2 \sqrt{\omega}}  \right)
\end{align*}

Plugin the result so that
\begin{align*}
    \cost{\X}(\centroids) &\leq \left( 2\beta_{priv} + 4\beta_0 + 4 \beta_0 \beta_{priv} \right)\opt_{k, \X} \\
    &+ 2(\beta_0 + 1)S\lambda_{priv} \\
    &+ O\left( (\beta_0 + 1)m^2 \left(\frac{n k^{1.5S}}{\sqrt{\omega M}} +  \frac{S(k-1)k^{1.5S}\sqrt{\log (1/\delta)}}{\epsilon_2 \sqrt{\omega}}  \right) \right)
\end{align*}

\end{proof}
\section{Addtional Experiment Information}
\subsection{DPLSF Details}

To approximate the effect of efficiently decodable net, the authors in \cite{chang2021locally} proposed a heuristic solution based on \emph{locality sensitive hashing forest}.
In the central DP implementation~\cite{googledpclustering}, they use the \emph{SmiHash}~\cite{charikar2002simhash} as the locality sensitive hash (LSH) function.
All the records of a data party are hashed by $L$ LSH functions and transferred to $L$-bit hash strings.
By the property of the LSH function, the records similar to each other have higher chance to have the same hash output. 

A trie (prefix tree) is built based on the hashed points level by level from the root, which is also called LSH tree.
The root node contains all data points.
The counts of the data points are calculated by the Laplace mechanism with the sensitivity 1 and stored within the node.
A node at level $l$ will be branched into two new nodes at level $l+1$ if it has noisy count larger than $3\theta$.
The data points in the node will be partitioned to either of the children node based on their hash outputs' $l+1$-th bit.
Because the hashes for each data point has length $L$, the height of the trie is at most $L$.

After building the LSH trie, the means of the data points assigned to the leaves are calculated by 1) computing the noisy sum of data points in a leave by Gaussian mechanism; 2) dividing the noisy sum by the noisy count stored in the leaves.
After such operation, the leaves with the privatized mean and privatized count form a differentially private coreset of the dataset.
Finally, any favorable \kmeans algorithm can be used on the private coreset to generate the final $k$ centers.
In the implementation, $L=20$ and $\theta=\min \left\{ 10 \sigma \sqrt{m}, \lfloor\frac{n}{2k}\rfloor \right\}$, where the $\sigma$ is the variance of Gaussian noise when calculating the means of the data points.
The intuition for $\theta$ is that with such parameter setting, the error of the average is expect to around $0.1$.

\subsection{\ldpca Details}
\label{app:ldp}
When $k' \leq 3e^{\epsilon_2} + 2$, data parties use \grr$: [k'] \rightarrow [k']$.
Denote the randomized output of the \grr as $v\rel{\id}{\ell}$.
The output of \grr can be characterize as
\begin{align*}
    \Pr{v\rel{i}{\ell} = a} = \begin{cases}
    p=\frac{e^{\epsilon_2}}{e^{\epsilon_2} + k' - 1}, \text{ if } a = \phi_{\centroids\rel{}{\ell}}(x\rel{i}{\ell}) \\
    q=\frac{1}{e^{\epsilon_2} + k' - 1}, \text{ if } a \neq \phi_{\centroids\rel{}{\ell}}(x\rel{i}{\ell}) \\
    \end{cases}
\end{align*}
The membership encoding information sent to the central server is the perturbed values $\mathbf{I}\rel{}{\ell}=\left[ v\rel{1}{\ell}, \ldots, v\rel{n}{\ell}\right]$.

When $k' > 3e^{\epsilon_2} + 2$, data parties use \olh.
Each user is coupled with a hash function $H\rel{\zeta}{\ell}: [k'] \rightarrow [\lfloor e^{\epsilon_2} + 1 \rfloor]$.
The hash outputs $H\rel{\zeta_\id}{\ell}\left(\phi_{\centroids\rel{}{\ell}}(x\rel{\id}{\ell}) \right)$ are randomized with GRR and sent as tuples $\mathbf{I}\rel{}{\ell}=\left[ ( H\rel{\zeta_1}{\ell}, v\rel{1}{\ell}), \ldots, ( H\rel{\zeta_n}{\ell}, v\rel{n}{\ell}) \right]$ to the central server.
With OLH, the randomized result can support the true label with probability $p=\frac{1}{2}$, but will support a wrong partition with $q=\frac{1}{e^{\epsilon_2}+1}$.

After receiving the membership encoding information from the data parties, the server can instantiate the \weightestimate as the following.
It has a tuple for each user $(v\rel{i}{1}, \ldots, v\rel{i}{S})$ if using \grr, or  $(H\rel{i}{1}, \ldots,  H\rel{i}{S}, v\rel{i}{1}, \ldots, v\rel{i}{S})$ if using \olh.
For a grid node $\grid_{(a_1, \ldots, a_S)}$, the support of this node is defined as 
\begin{align*}
    \support\left(\grid_{(a_1, \ldots, a_S)}\right) = \begin{cases}
    \{i \mid \forall \ell \in [S], v\rel{i}{\ell} = a_\ell\}, \text{ for \grr} \\
    \{i \mid \forall \ell \in [S], v\rel{i}{\ell} = H\rel{\zeta_i}{\ell}(a_\ell) \}, \text{ for \olh}
    \end{cases}
\end{align*}
If $\Delta_{HD}$ be the hamming distance between the local centers lists assembling two grid nodes $g = \grid_{(a_1, \ldots, a_S)}$ and $g'=\grid_{(a'_1, \ldots, a'_S)}$, the probability of a data point should be assigned to $g$ but perturbed to support $g'$ is
\begin{align*}
    \Pr{i \in \support(g') \mid \forall \ell \in [S], \phi_{\centroids\rel{}{\ell}}(x\rel{i}{\ell}) =a_{\ell} } = p^{S-\Delta_{HD}}q^{\Delta_{HD}}
\end{align*}
With the probabilities between all $g$ and $g'$, the central server can build a probability transition matrix $\mathbf{P}$, where each column represents the real assignment and each row represents a possible perturbed assignment.
Let $\hat{w}$ as the count vector for the supports observed.
An unbiased estimates for the intersection cardinality vector $w(\grid)$ can be calculated by solving a linear equation $ w(\grid)= \mathbf{P}^{-1} \hat{w}$.

Based on the known LDP protocol error analysis (Proposition 10 in \cite{wang2019answering}), the variance of a estimated weight is in the order $O(\frac{n}{\epsilon_2^{2S}})$.
This \ldpca introduces large randomness when $\epsilon_2$ is small, or the number of data parties is large.

\subsection{More Ablation Study Results}

\mypara{Intersection cardinality accuracy comparison.} 
The Figure~\ref{fig:ca-estimate-error-extra} shows the intersection cardinality accuracy comparison on the Taxi and Letter dataset.
As we can see, the \dpfmpsimprove produces lower cardinality relative error on the Taxi dataset for most $\epsilon$ values.
However, the \indca methods has smaller error when $\epsilon < 2$ for the Letter dataset as shown in Figure~\ref{fig:letter_ca_2way_S=2} and \ref{fig:letter_ca_2way_S=4}.
The first reason is that the attributes of the Letter is relatively independent to each other in the sense that the attributions' Pearson correlation is more closer to 0.
The second reason is that the number of records of Letter dataset is smaller than Taxi or Loan, so the relative errors of the sketch approach has larger impact over the true cardinalities.
\begin{figure*}
     \centering
     \begin{subfigure}[b]{0.5\textwidth}
        \centering
        \includegraphics[width=\textwidth]{figure/intersection/intersection-legend-line.png}
    \end{subfigure}
    \\
    \begin{subfigure}[b]{0.24\textwidth}
    \centering
        \includegraphics[width=\textwidth]{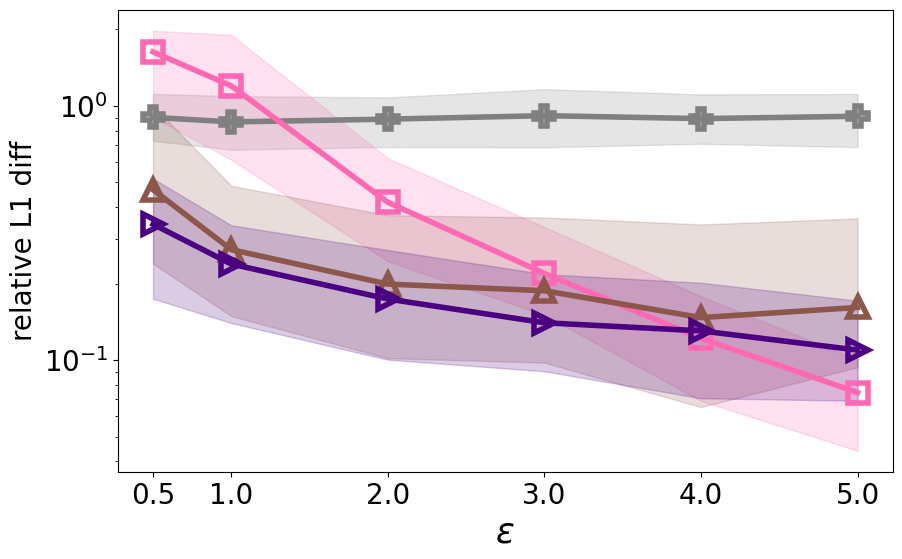}
        \caption{Taxi $S=2$}
        \label{fig:taxi_ca_2way_S=2}
    \end{subfigure}
    \begin{subfigure}[b]{0.24\textwidth}
        \centering
        \includegraphics[width=\textwidth]{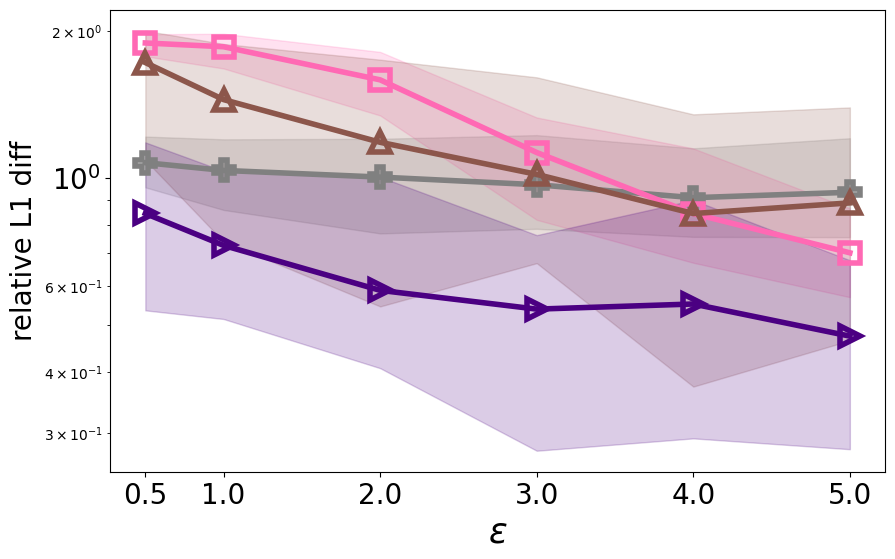}
        \caption{Taxi $S=4$}
        \label{fig:taxi_ca_2way_S=4}
     \end{subfigure}
    \begin{subfigure}[b]{0.24\textwidth}
        \centering
        \includegraphics[width=\textwidth]{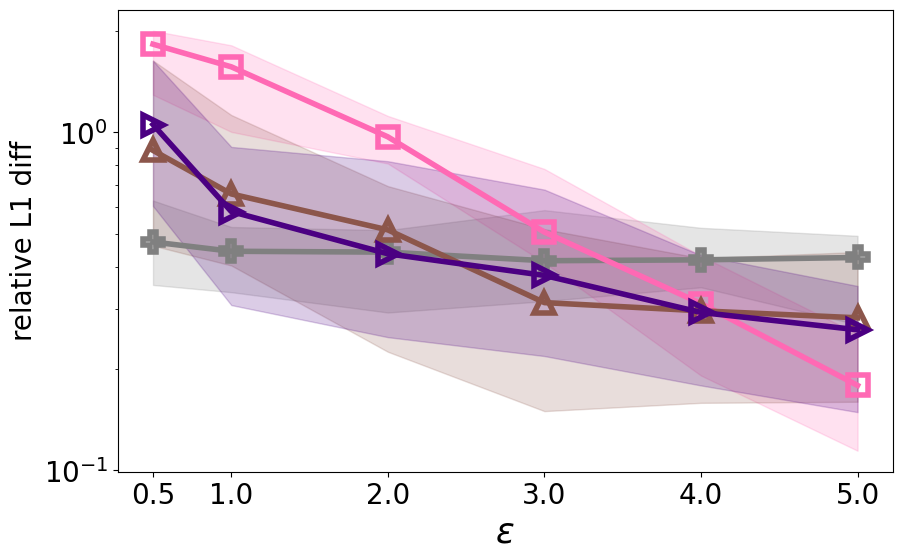}
        \caption{Letter $S=2$}
        \label{fig:letter_ca_2way_S=2}
    \end{subfigure}
    \begin{subfigure}[b]{0.24\textwidth}
        \centering
        \includegraphics[width=\textwidth]{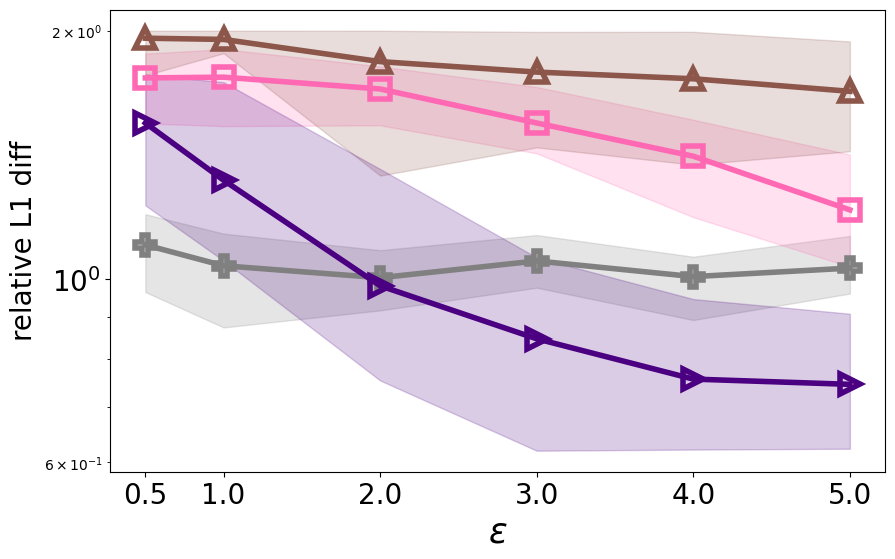}
        \caption{Letter $S=4$}
        \label{fig:letter_ca_2way_S=4}
    \end{subfigure}
    \vspace{-0.4cm}
     \caption{relative error of estimate the intersection cardinalities}
     \label{fig:ca-estimate-error-extra}
\end{figure*}

\mypara{Impact of the private local clustering or intersection cardinality to the final loss.}
The additional results on the Taxi and Letter datasets are shown in Figure~\ref{fig:break-error-extra}. 
Similar to the results in Figure~\ref{fig:break-error}, enforcing the \localclustering private but using the non-private intersection cardinality estimation gives the cost closer to the end-to-end private ones when $S=2$;
but the errors introduced by adapting a privacy-preserving solution for either component have a similar effect on the final loss when $S=4$.

\begin{figure*}
     \centering
     \begin{subfigure}[b]{0.7\textwidth}
        \centering
        \includegraphics[width=\textwidth]{figure/break/break-loss-legend-line.png}
    \end{subfigure}
    \\
    \begin{subfigure}[b]{0.24\textwidth}
    \centering
        \includegraphics[width=\textwidth]{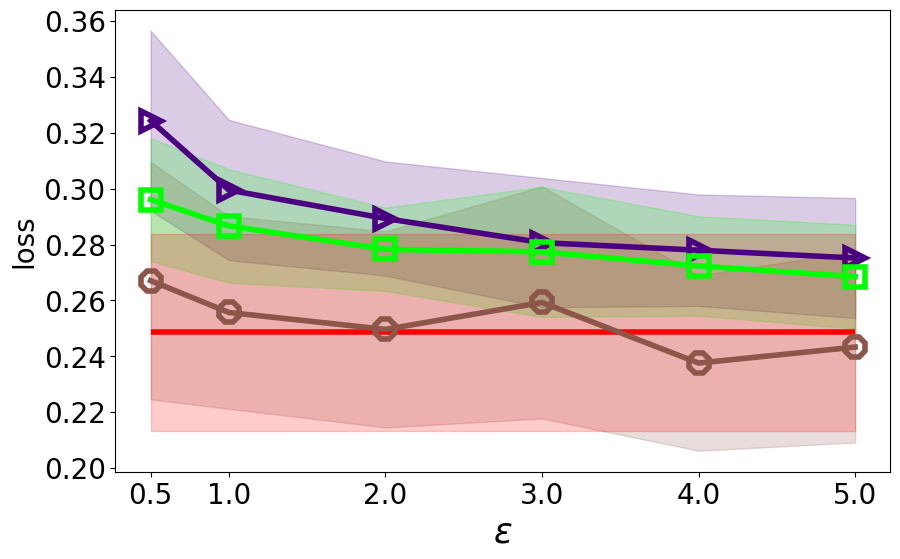}
        \caption{Taxi $S=2$}
        \label{fig:taxi_break_S=2}
    \end{subfigure}
    \begin{subfigure}[b]{0.24\textwidth}
        \centering
        \includegraphics[width=\textwidth]{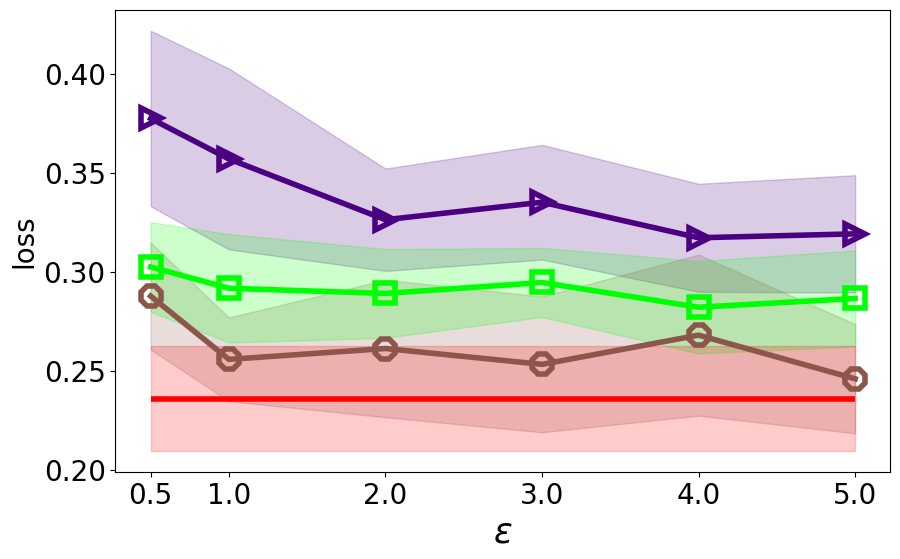}
        \caption{Taxi $S=4$}
        \label{fig:taxi_break_S=4}
    \end{subfigure}
    \begin{subfigure}[b]{0.24\textwidth}
        \centering
        \includegraphics[width=\textwidth]{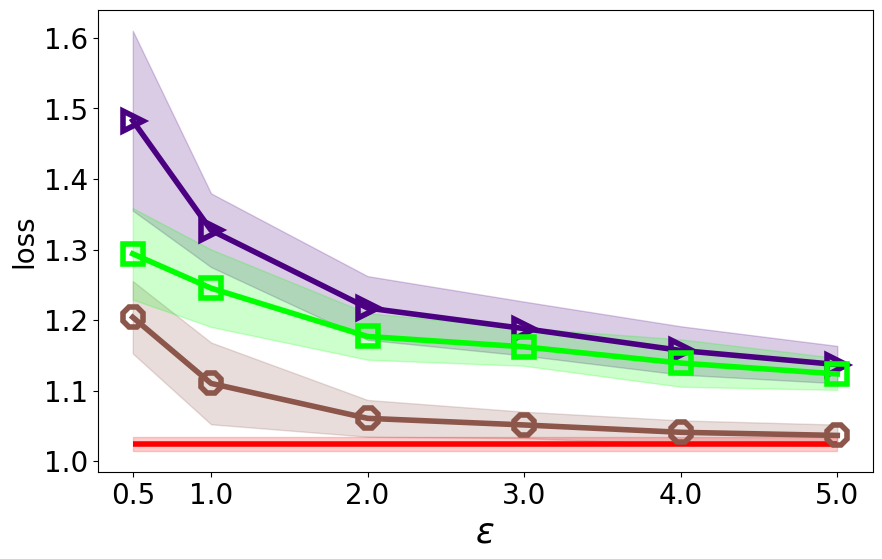}
        \caption{Letter $S=2$}
        \label{fig:letter_break_S=2}
    \end{subfigure}
    \begin{subfigure}[b]{0.24\textwidth}
        \centering
        \includegraphics[width=\textwidth]{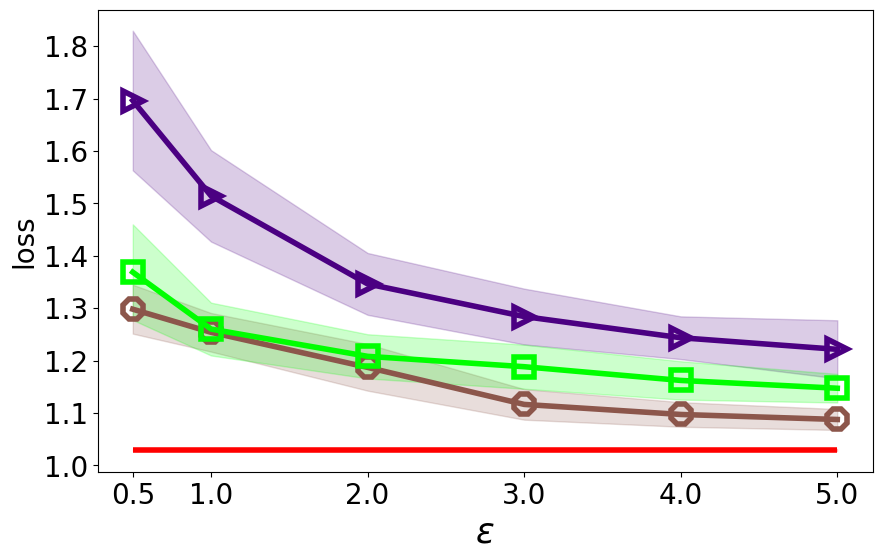}
        \caption{Letter $S=4$}
        \label{fig:letter_break_S=4}
    \end{subfigure}
    \vspace{-0.4cm}
     \caption{Comparing the impact of (enforcing privacy on) different components }
     \label{fig:break-error-extra}
\end{figure*}

\mypara{Extra local $k'$ experiments}
Figure~\ref{fig:localk-extra} shows additional experiments on Taxi, Loan and letter datasets of varying $k'$.
The trends are similar to Figure~\ref{fig:localk} on the synthetic dataset.
\begin{figure*}
    \centering
    \begin{subfigure}[b]{0.3\textwidth}
        \centering
        \includegraphics[width=\textwidth]{figure/localk/localk-loss-legend-line.png}
    \end{subfigure}
    \\
     \begin{subfigure}[b]{0.24\textwidth}
         \centering
         \includegraphics[width=\textwidth]{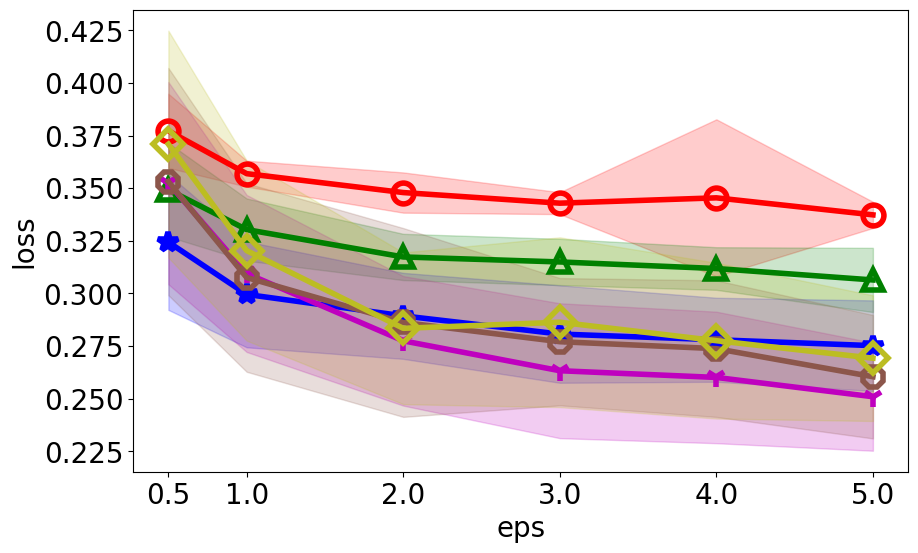}
         \vspace{-0.5cm}
         \caption{Taxi, $S=2$}
         \label{fig:taxi_localk2}
     \end{subfigure}
     \begin{subfigure}[b]{0.24\textwidth}
         \centering
         \includegraphics[width=\textwidth]{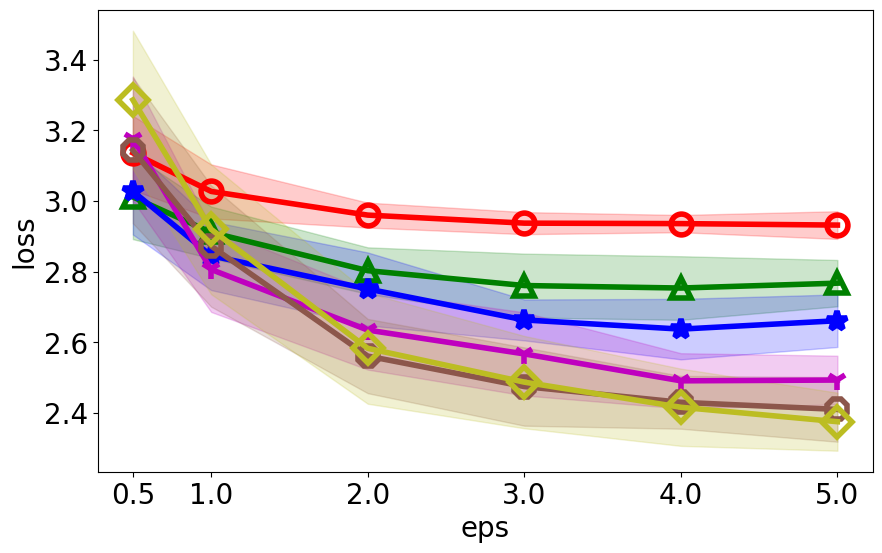}
         \vspace{-0.5cm}
         \caption{Loan, $S=2$}
         \label{fig:loan_localk2}
     \end{subfigure}
     \begin{subfigure}[b]{0.24\textwidth}
         \centering
         \includegraphics[width=\textwidth]{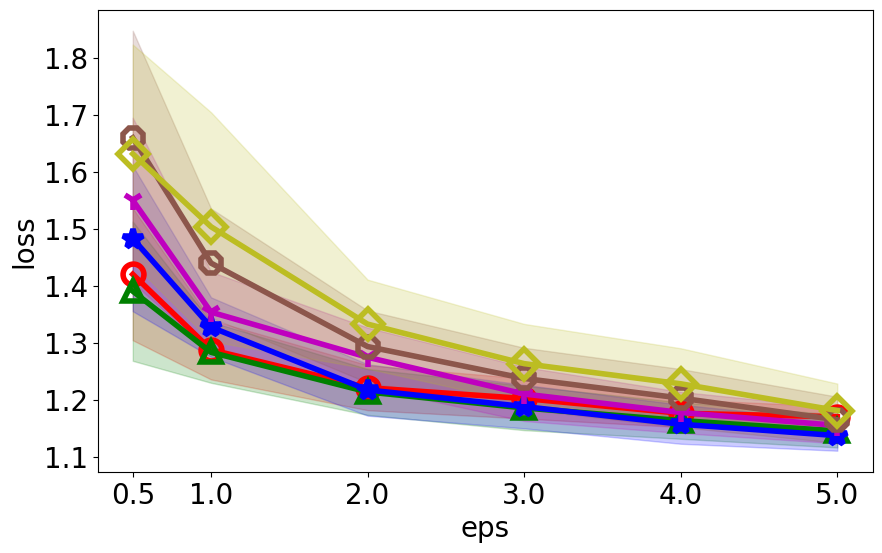}
         \vspace{-0.5cm}
         \caption{Letter, $S=2$}
         \label{fig:letter_localk2}
     \end{subfigure}
     \\
     \begin{subfigure}[b]{0.24\textwidth}
         \centering
         \includegraphics[width=\textwidth]{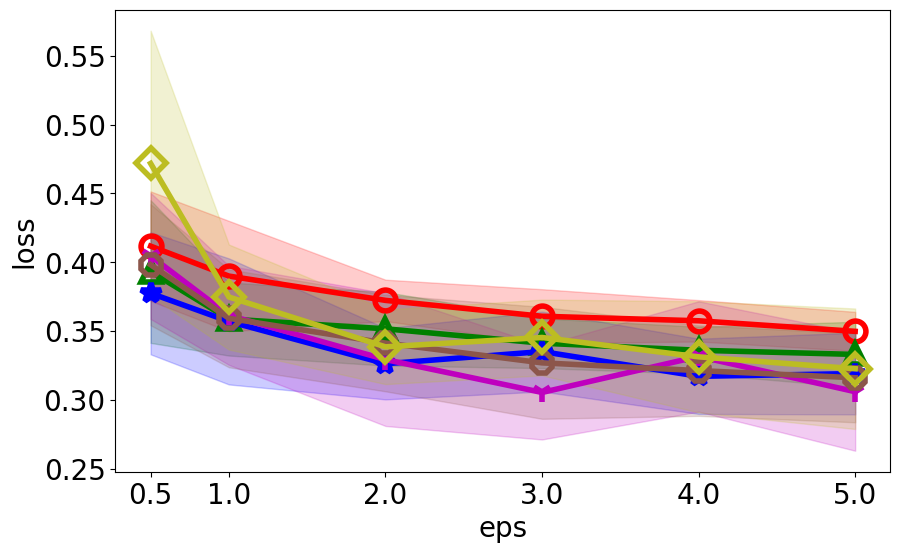}
         \vspace{-0.5cm}
         \caption{Taxi, $S=4$}
         \label{fig:taxi_localk4}
     \end{subfigure}
     \begin{subfigure}[b]{0.24\textwidth}
         \centering
         \includegraphics[width=\textwidth]{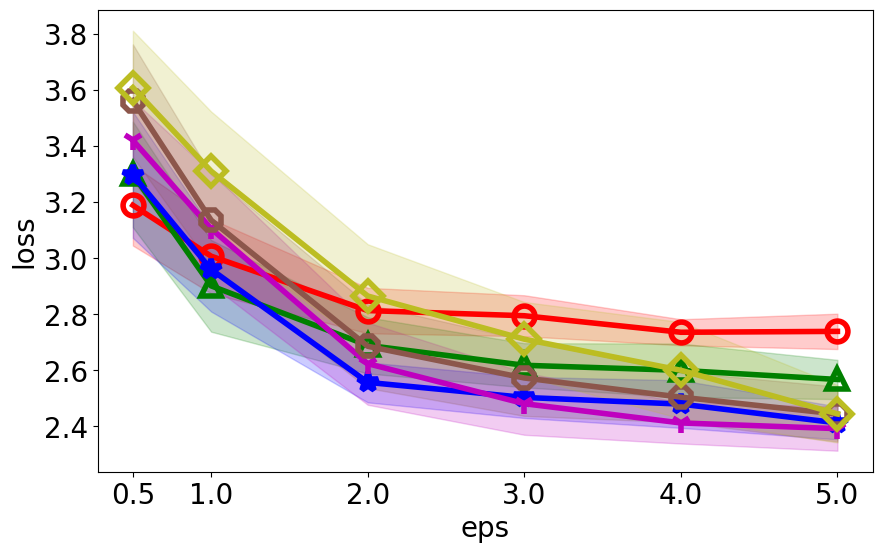}
         \vspace{-0.5cm}
         \caption{Loan, $S=4$}
         \label{fig:loan_localk4}
     \end{subfigure}
     \begin{subfigure}[b]{0.24\textwidth}
         \centering
         \includegraphics[width=\textwidth]{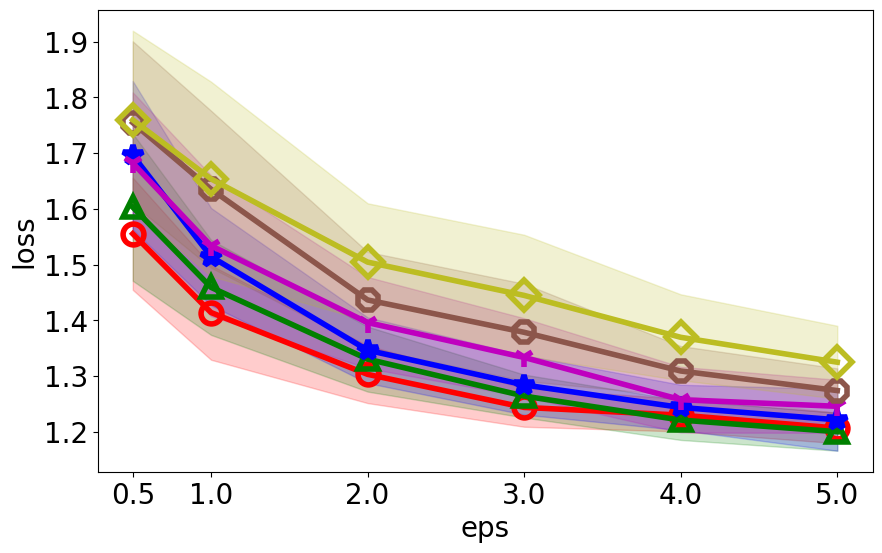}
         \vspace{-0.5cm}
         \caption{Letter, $S=4$}
         \label{fig:letter_localk4}
     \end{subfigure}
    \vspace{-0.3cm}
    \caption{Different local $k'$ with different privacy budget on real world datasets}
    \label{fig:localk-extra}
\end{figure*}

\begin{figure*}
    \centering
    \begin{subfigure}[b]{0.35\textwidth}
        \centering
        \includegraphics[width=\textwidth]{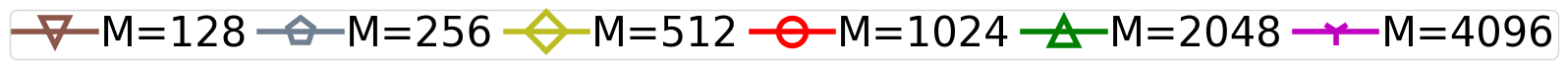}
    \end{subfigure}\\
    \begin{subfigure}[b]{0.23\textwidth}
        \centering
        \includegraphics[width=\textwidth]{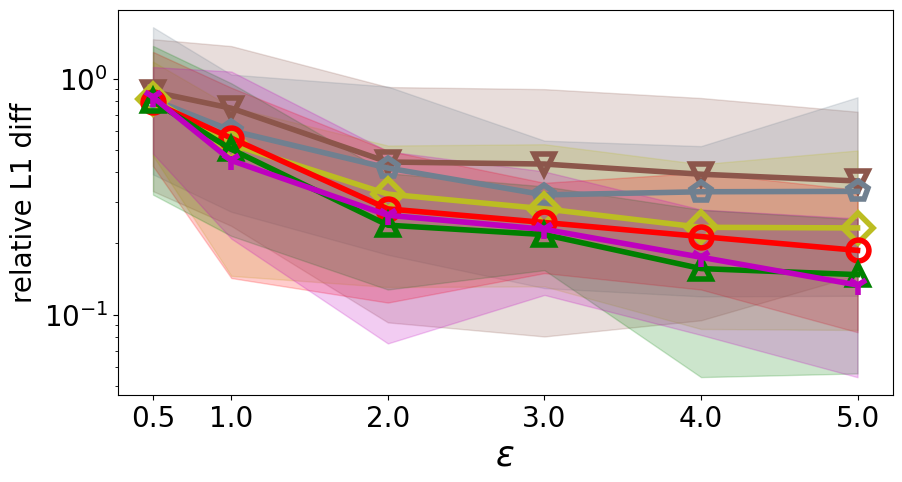}
        \vspace{-0.6cm}
        \caption{Mixed Gaussian $S=2$}
        \label{fig:vary_M_S=2}
    \end{subfigure}
    \begin{subfigure}[b]{0.23\textwidth}
        \centering
        \includegraphics[width=\textwidth]{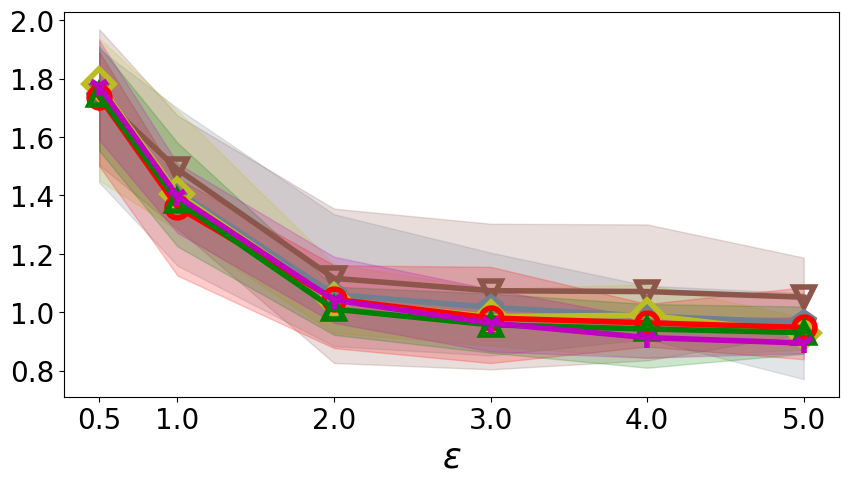}
        \vspace{-0.6cm}
        \caption{Mixed Gaussian $S=4$}
        \label{fig:vary_M_S=4}
    \end{subfigure}
    \vspace{-0.4cm}
    \caption{Intersection error with different FM repetition $M$}
    \label{fig:vary_M}
    \vspace{-0.4cm}
\end{figure*}
\mypara{Impact of FM sketch repetition $M$.}
On the other hand, we show in Figure~\ref{fig:vary_M_S=2} that the larger $M$ can provide more accurate weight estimates (intersection cardinalities) when $M \leq 2048$.
However, the benefit will vanish when the $M$ is large enough, i.e., comparing $M = 2048$ and $4096$.

\begin{table}
\resizebox{0.95\columnwidth}{!}{%
\begin{tabular}{c|c|c|c|c|c}
\toprule
\multirow{2}{*}{$\epsilon$}  &\multirow{2}{*}{method} & \multicolumn{2}{c|}{ \kmeans loss} & \multicolumn{2}{c}{V-measure score} \\
& &$S=2$ & $S=4$ &$S=2$ & $S=4$\\
\midrule
\multirow{2}{*}{$\infty$} 
& VFL non-priv & 0.0763 & 0.0772 & 0.9781 & 0.9989\\
& central non-priv & \multicolumn{2}{c|}{0.0763} &  \multicolumn{2}{c}{1.0} \\
\midrule
\multirow{4}{*}{$1$}  
& \dpfmpsimprove &  0.7193 & 1.1502 & 0.9441 &  0.8771  \\
& \indca & 1.1724 & 1.5676 & 0.8752 & 0.8210\\
& \ldpimprove & 1.2716 & 1.9898 & 0.8638 & 0.6716\\
& DPLSF & \multicolumn{2}{c|}{0.1595} & \multicolumn{2}{c}{0.9945}\\
\midrule
\multirow{4}{*}{$4$}  
& \dpfmpsimprove &  0.1525 & 0.4016 & 0.9850 & 0.9868\\
& \indca & 0.9640 & 1.3083 & 0.9106 & 0.9206\\
& \ldpimprove & 0.3069 & 0.9945 & 0.9841 &  0.9440\\
& DPLSF & \multicolumn{2}{c|}{0.1240} &  \multicolumn{2}{c}{0.9999}\\
\bottomrule
\end{tabular}
}
\caption{Detailed experimental results on Mixed Gaussian dataset with different metrics.}
\vspace{-0.5cm}
\label{table:cross metric}
\end{table}


\subsection{Dataset Details}
The details of the real-world datasets are listed in Table~\ref{table:attributes}.
The Silhouette of different datasets are shown in Figure~\ref{fig:silhouette}.
Because both Taxi and Loan have peak around $k=5$, we pick $k=5$ for the experiments in this paper.

\paragraph{Partition of Taxi dataset.}
\begin{sloppypar}
When we consider $S=2$, we partition the Taxi data attributes with \path{[pickup_datetime, pickup_longitude, pickup_latitude, passenger_count, ]} and \path{[dropoff_datetime, dropoff_longitude, dropoff_latitude, trip_duration]};
when $S=4$, we partition the it with \path{[ dropoff_longitude, pickup_longitude], [passenger_count, dropoff_datetime],
[pickup_latitude, dropoff_latitude]}, and \path{ [pickup_datetime, trip_duration]}.
\end{sloppypar}

\paragraph{Partition of Letter dataset.}
When we consider $S=2$, we partition the Letter dataset according to attribute indices: [1, 3, 6, 7, 9, 12, 14, 15] and [2, 4, 5, 8, 10, 11, 13, 16]; 
When $S=4$, we partition as the following: [1, 6, 9, 14],
[2, 5, 10, 16],
[3, 7, 12, 15],
and [4, 8, 11, 13].

\begin{table}
\begin{tabular}{p{2cm}|p{5.5cm}}
 \toprule
Dataset & Attributes \\
\hline
New York Taxi & \path{pickup_datetime, 
    dropoff_datetime, 
    passenger_count, 
    trip_duration,
    pickup_longitude, 
    pickup_latitude,  
    dropoff_longitude,
     dropoff_latitude} \\
\hline
Loan & \path{AMT_CREDIT,
    CNT_FAM_MEMBERS,
    ENTRANCES_AVG,
    DEF_30_CNT_SOCIAL_CIRCLE,
    APARTMENTS_AVG,
    OBS_30_CNT_SOCIAL_CIRCLE,
    AMT_GOODS_PRICE,
    FLOORSMAX_AVG,
    DEF_60_CNT_SOCIAL_CIRCLE,
    LIVINGAREA_AVG,
    OBS_60_CNT_SOCIAL_CIRCLE,
    CNT_CHILDREN,
    LIVINGAPARTMENTS_AVG,
    AMT_ANNUITY,
    ELEVATORS_AVG,
    COMMONAREA_AVG} \\
\hline
Letter & all attributes except the label \\
\bottomrule
\end{tabular}
\caption{Details of the datasets.}
\label{table:attributes}
\end{table}
\begin{figure*}
    \centering
    \centering
     \begin{subfigure}[b]{0.3\textwidth}
         \centering
         \includegraphics[width=\textwidth]{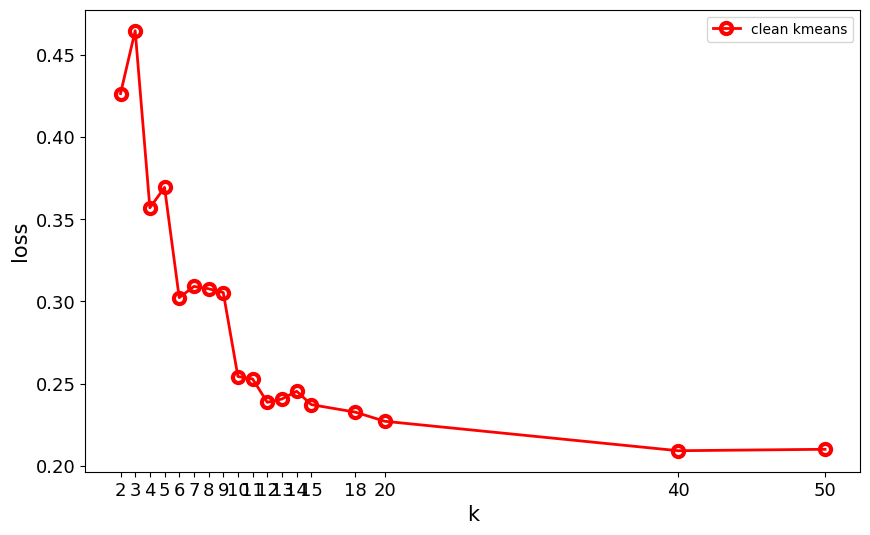}
         \caption{Taxi}
         \label{fig:silhouette_taxi}
     \end{subfigure}
     \begin{subfigure}[b]{0.3\textwidth}
         \centering
         \includegraphics[width=\textwidth]{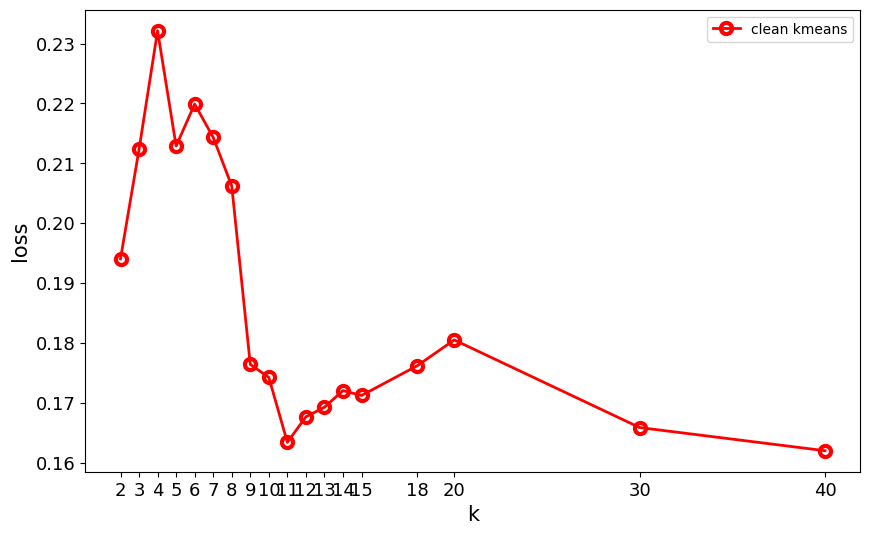}
         \caption{Loan }
         \label{fig:silhouette_loan}
     \end{subfigure}
     \begin{subfigure}[b]{0.3\textwidth}
         \centering
         \includegraphics[width=\textwidth]{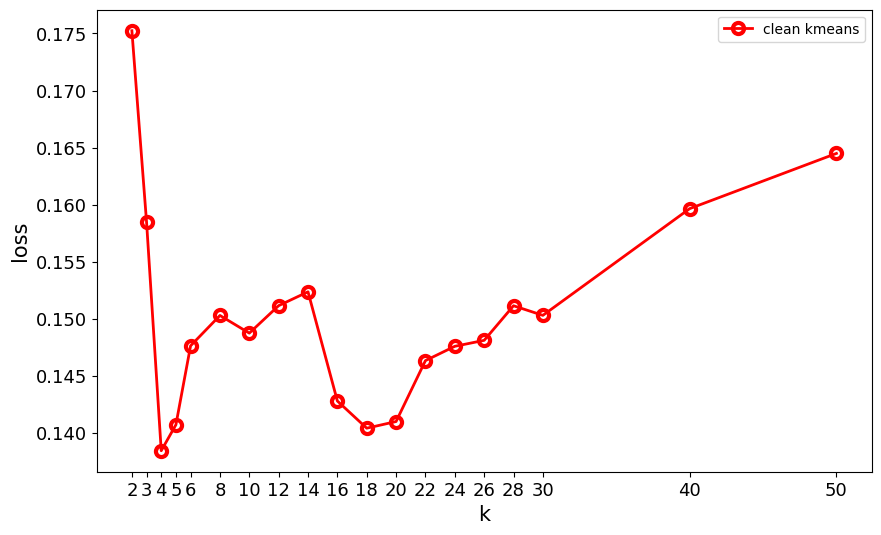}
         \caption{Letter}
         \label{fig:silhouette_letter}
     \end{subfigure}
    \caption{Silhouette score for reasoning $k$ for experiments.
    }
    \label{fig:silhouette}
\end{figure*}

\end{document}